\let\latexaddtocontents\addtocontents
\documentclass[a4paper,twocolumn,allowtoday,accepted=2025-04-07]{quantumarticle}
\pdfoutput=1 
\usepackage[colorlinks=true,linkcolor=blue,citecolor=blue,urlcolor=blue]{hyperref}
\let\addtocontents\latexaddtocontents

\usepackage[utf8]{inputenc}
\usepackage[german, english]{babel}
\usepackage[dvipsnames]{xcolor}

\usepackage[backend=bibtex,sorting=none,firstinits=true,style=phys,maxbibnames=5,biblabel=brackets,chaptertitle=false,pageranges=false,doi=false]{biblatex}
\usepackage{dcolumn}  
\usepackage{bm}        
\usepackage{bbold}
\usepackage{amsthm}
\usepackage{mathtools}
\usepackage{tikz}
\usetikzlibrary{shapes.geometric, arrows, shapes.multipart}
\usepackage{booktabs, tabularx}

\usepackage{blkarray}
\usepackage{xfrac}

\usepackage[export]{adjustbox}
\usepackage{titlesec}
\usepackage{amsfonts}
\usepackage{amssymb}
\usepackage{braket}

\usepackage{subfigure}

\newcommand{\ab}{{\bf{a},\bf{b}}}

\newcommand{\id}{\mathbb{1}}

\newtheorem{Theorem}{Theorem}

\newtheorem{Lemma}{Lemma} 

\newtheorem{Def}{Definition}

\newtheorem{Conj}{Conjecture}

\newtheorem{Prop}{Proposition}

\begin{document}

\title{Non-onsite symmetries and quantum teleportation in split-index matrix product states}

\date{\today}

\author{David T. Stephen}
\affiliation{Department of Physics and Center for Theory of Quantum Matter, University of Colorado Boulder, Boulder, CO, USA}
\affiliation{Department of Physics and Institute for Quantum Information and Matter, Caltech, Pasadena, CA, USA}

\begin{abstract}
We describe a class of spin chains with new physical and computational properties. On the physical side, the spin chains give examples of symmetry-protected topological phases that are defined by non-onsite symmetries, \text{i.e.}, symmetries that are not a tensor product of single-site operators. These phases can be detected by string-order parameters, but notably do not exhibit entanglement spectrum degeneracy. On the computational side, the spin chains represent a new class of states that can be used to deterministically teleport information across long distances,
with the novel property that the necessary classical side processing is a non-linear function of the measurement outcomes. We also give examples of states that can serve as universal resources for measurement-based quantum computation, providing the first examples of such resources without entanglement spectrum degeneracy. 
The key tool in our analysis is a new kind of tensor network representation which we call split-index matrix product states (SIMPS). We develop the basic formalism of SIMPS, compare them to matrix product states, show how they are better equipped to describe certain kinds of non-onsite symmetries including anomalous symmetries, and discuss how they are also well-suited to describing quantum teleportation and constrained spin chains. 
\end{abstract}

\maketitle

\section{Introduction} \label{sec:intro}

Quantum spin chains are an endless well of surprising physics. Models such as the Affleck-Kennedy-Lieb-Tasaki (AKLT) chain \cite{Affleck1989}, which initially appear featureless as they break no symmetries, turn out to have rich physics including degeneracy in the boundary and entanglement spectrum, hidden symmetry breaking, and non-local string order \cite{Kennedy1992,denNijs1989,Pollmann2010,Pollmann2012a}. These properties are a consequence of the presence of symmetry-protected topological (SPT) order \cite{Pollmann2010,Chen2011,Schuch2011,Pollmann2012}. Remarkably, these same spin chains are also useful for quantum information processing. For example, they can be used to teleport information across large distances \cite{Popp2005,Else2012,Wahl2012}, generate cat-like entanglement \cite{Tantivasadakarn2021,Williamson2021}, and even perform universal measurement-based quantum computation (MBQC) \cite{Raussendorf2003,Miyake2010,Stephen2017}.

The computational and physical properties of quantum spin chains are intimately related. It has been proven that SPT order---or alternatively string order---implies quantum computational utility of the ground states of spin chains \cite{Doherty2009,Else2012,Miller2015,Marvian2017,Stephen2017,Raussendorf2023}. Conversely, the ability to perform teleportation has been used to witness phase transitions \cite{Verstraete2004,Verstraete2004a,Popp2005,Skrovseth2009} and to probe for SPT order in experiment \cite{Azses2020,Smith2023}.
These relations can be most clearly seen using the matrix product state (MPS) formalism \cite{Perez-Garcia2006,Cirac2021}. Therein, the many-body wavefunction is described in terms of many local tensors whose symmetries underlie both the computational and physical properties of the MPS \cite{Else2012,Pollmann2012}. More generally, the MPS formalism has established itself as the ultimate tool for studying 1D spin chains both analytically and numerically \cite{Verstraete2008review,SCHOLLWOCK2011,Orus2014,Bridgeman2017,Cirac2021}. Therefore, any extension of MPS can potentially have impact on a wide range of fields and applications.

In this paper, we construct a class of quantum spin chains that have novel physical and computational properties. The key concept that we introduce is a new kind of tensor network representation we call \textit{split-index matrix product states (SIMPS)} \eqref{eq:simps}. While equivalent to MPS in terms of which states can be efficiently captured, SIMPS offer an alternative description that we argue is essential to understanding the properties of the constructed models. For example, a key physical property of the models is that they possess symmetries which are non-onsite, \text{i.e.}, not a tensor product of single-site unitaries. Compared to MPS, we show that SIMPS provide a much simpler way to identify and analyze such symmetries. Similarly, SIMPS offer a clearer picture of the computational capabilities of certain models, including ones that have previously appeared in the literature. We also show that SIMPS are well-suited to capturing constrained spin chains, including models exhibiting quantum many-body scars and Hilbert space fragmentation \cite{Moudgalya2022}.

Overall, the SIMPS formalism is an extension of the MPS formalism that better captures certain physical and computational properties of spin chains, clarifies some existing results in the literature, and enables the construction of new models. We summarize the key results of this paper in the next section.

\subsection{Summary of results}

This paper is organized as follows:

\begin{itemize}
    \item In Sec.~\ref{sec:simps}, we introduce the SIMPS formalism, determine when an MPS is suited to a SIMPS representation, show how to optimally convert between MPS and SIMPS representations, and prove a number of general results including a fundamental theorem of SIMPS.
    
    \item In Sec.~\ref{sec:spt}, we use SIMPS to construct a large class of spin chains that have SPT order under non-onsite symmetries. We characterize these phases in terms of string order, response to flux insertion, (a lack of) entanglement spectrum degeneracy, and by revealing the ``universal fingerprints'' in the entanglement structure. We also show how SIMPS are well-equipped to handle anomalous symmetries.

    \item In Sec.~\ref{sec:computation}, we show how these SIMPS are resources for long-range quantum teleportation and MBQC. Key differences distinguishing our resources from existing ones are the necessity of non-linear classical side processing, and MBQC universality in the absence of entanglement spectrum degeneracy.

    \item In Sec.~\ref{sec:constrained}, we show how SIMPS are suited to describing constrained spin chains, since generic local constraints can be encoded as linear constraints on the SIMPS tensors rather than non-linear constraints on an MPS tensor. We show that previous MPS Ans\"{a}tze used to parameterize constrained Hilbert spaces arise naturally from SIMPS, which then provides a route for more general Ans\"{a}tze.

    \item Finally, in Sec.~\ref{sec:discussion}, we apply our results to clarify an example that previously appeared in the literature, discuss how our results support the conjectured equivalence between quantum teleportation and SPT order, and outline possible future applications of SIMPS.
\end{itemize}

\section{The SIMPS representation} \label{sec:simps}

In this section, we define the SIMPS representation, prove some fundamental properties of SIMPS, and show how to convert between MPS and SIMPS. To begin, we review the definition of MPS. An MPS with periodic boundary conditions (PBC) has the following form,
\begin{equation} \label{eq:mps}
    |\psi\rangle = \sum_{i_1,\dots, i_N} \mathrm{Tr}(A^{i_1}A^{i_2}\cdots A^{i_N}) |i_1i_2\cdots i_N\rangle,
\end{equation}
where the physical indices $i_k$ takes values from in $0,\dots, d-1$, each $A^i$ is a $D\times D$ matrix that acts in the so-called \textit{virtual space} and $D$ is the \textit{bond dimension}. An important class of MPS is that of \textit{normal MPS},

\begin{Def}[Normal MPS] \label{def:normalmps}
    An MPS tensor $A^i$ is said to be \emph{normal} if there is a number $L_0$ called the \emph{injectivity length} such that $\mathrm{span}\{A^{i_1}A^{i_2}\cdots A^{i_{L_0}}\}_{i_1,\dots, i_{L_0}}=M_D$ where $M_D$ is the space of all complex $D\times D$ matrices. If $L_0=1$, the MPS is said to be \emph{injective}.
\end{Def}

\noindent
Physically, a normal MPS is the unique ground state of a particular gapped parent Hamiltonian, and therefore has a finite correlation length \cite{Perez-Garcia2006}. Conversely, non-normal MPS have long-range correlations and are associated with spontaneous symmetry breaking \cite{Schuch2010}.
Note that, if the spanning property holds for length $L_0$, it also holds for all $L\geq L_0$. This property is equivalent to the map 
\begin{equation} \label{eq:inj_map}
 A_L (X) = \sum_{i_1\dots i_{L_0}}\mathrm{Tr}(A^{i_1}A^{i_2}\cdots A^{i_{L}} X)|i_1\cdots i_{L}\rangle   
\end{equation}
being an injective map from the virtual degrees of freedom to the physical degrees of freedom for all $L\geq L_0$ \cite{Perez-Garcia2006}. Thus injectivity means we can invert this map, allowing us to reveal the virtual legs of the MPS. This is important for constructing parent Hamiltonians \cite{Schuch2010} and string order parameters \cite{Pollmann2012a}.

Given the MPS $|\psi\rangle$, we can define a state on open boundary conditions (OBC) in the following way,
\begin{equation} \label{eq:mps_obc}
    |\bar{\psi}\rangle = \sum_{\substack{i_1,\dots, i_N \\ a,b}} \langle a|A^{i_1}A^{i_2}\cdots A^{i_N}|b\rangle |a,i_1\cdots i_N,b\rangle,
 \end{equation}
where we have added $D$-dimensional qudits on the edges of the chain indexed by $a$ and $b$ which take values in $0,\dots, D-1$, such that $|\bar{\psi}\rangle$ lives in a Hilbert space $\mathbb{C}^D\otimes (\mathbb{C}^d)^{\otimes N} \otimes \mathbb{C}^D$. We call the $d$-dimensional spins the bulk spins and the $D$-dimensional spins the boundary spins. This OBC form will be useful when we describe quantum teleportation using MPS.

Now we introduce SIMPS. A SIMPS describes a state $|\psi\rangle$ defined on a ring of $N$ spins of dimension $d$ as follows,
\begin{equation} \label{eq:simps}
    |\psi\rangle = \sum_{i_1,\dots, i_N} \mathrm{Tr}(B^{i_1i_2}B^{i_2i_3}\cdots B^{i_Ni_1}) |i_1i_2\cdots i_N\rangle.
\end{equation}
where each $B^{ij}$ is a $\chi^i\times\chi^j$ matrix. Most explicit examples that we consider in this paper have $\chi^i=\chi$ for all $i$, in which case we call $\chi$ the SIMPS bond dimension. We will see soon that $\chi$ is in general different from the bond dimension $D$ of the standard MPS representation. 

As in an MPS representation, a SIMPS is described in terms of a trace of a product of matrices that depend on the state of the physical spins. However, unlike the standard MPS representation, each matrix in a SIMPS is determined by the states of two neighbouring spins. The name ``split-index'' MPS comes from the fact that the state of every physical spin is ``split'' across two tensors $B^{ij}$. This splitting operation is defined with respect to a particular basis of the physical spins. Therefore, one important difference between MPS and SIMPS is that the SIMPS representation is strongly dependent on the choice of physical basis. In particular, the bond dimension can in general  depend on the basis of the physical spins $\{|i\rangle\}$. To be more general, one could consider tensors that depend on a larger number of physical indices. But, after sufficient enlargement of the unit cell, we can always return to the form of Eq.~\eqref{eq:simps}.

We remark that the SIMPS construction is very similar to double-line tensor networks \cite{Shukla2018,Jiang2017,Williamson2016mpo}, which share the feature that a given physical degree of freedom is used in more than one tensor in the network (see also \cite{Lami2022}). 
Here, we examine this kind of representation from a more fundamental standpoint.

 We can also define an OBC form of SIMPS in the case that $\chi^i=\chi$ for all $i$,
\begin{equation} \label{eq:simps_obc}
    |\bar{\psi}\rangle = \sum_{\substack{i_1,\dots, i_N \\ \alpha,\beta}} \langle \alpha|B^{i_1i_2}\cdots B^{i_{N-1}i_N}|\beta\rangle |\alpha,i_1\cdots i_N,\beta\rangle,
 \end{equation}
in which we have added $\chi$-dimensional spins to the left and right boundaries of the chain indexed by $\alpha$ and $\beta$ which take values in $0,\dots, \chi-1$. Note that the most natural boundary condition in this representation involves spins of dimension $\chi$, rather than the corresponding MPS bond dimension $D$.

As for MPS, we can identify an important subclass of SIMPS that we call normal SIMPS,
\begin{Def}[Normal SIMPS] \label{def:normalsimps}
    A SIMPS tensor $B^{ij}$ is \emph{normal} if there is a finite $L_1$ such that $\mathrm{span}\{B^{s_1i_1}B^{i_1i_2}\cdots B^{i_{L_1}s_2}\}_{i_1,\dots i_{L_1}}=M_{\chi^{s_1},\chi^{s_2}}$ for all $s_1,s_2$, where $M_{\chi^{s_1},\chi^{s_2}}$ is the space of all $\chi^{s_1}\times\chi^{s_2}$ matrices. If $L_1=1$, the SIMPS tensor is said to be \emph{injective}.
\end{Def}
\noindent

{The above definition is equivalent to the condition that the state obtained by projecting every $(L_1+1)$-th spin onto an arbitrary fixed basis state is an injective MPS.
As in the case of MPS, this property is stable in that, if it holds for some length $L$, then it holds for $L+1$ as well, and therefore for all lengths $>L$. To see this, observe that,
\begin{equation}
    \begin{aligned}
        &\mathrm{span}\{B^{s_1i_1}B^{i_1i_2}\cdots B^{i_{L+1}s_2}\}_{i_1,\dots, i_{L+1}} \\ 
        &= \mathrm{span}\{B^{s_1i_1}B^{i_1i_2}\cdots B^{i_{L}s_3}B^{s_3s_2}\}_{i_1,\dots, i_{L};s_3} \\
        &\supset \mathrm{span}\{B^{s_1i_1}B^{i_1i_2}\cdots B^{i_{L-1}s_3}B^{s_3s_2}\}_{i_1,\dots, i_{L-1};s_3} \\
        &= \mathrm{span}\{B^{s_1i_1}B^{i_1i_2}\cdots B^{i_{L}s_2}\}_{i_1,\dots, i_{L}}
    \end{aligned}
\end{equation}
where, going from the second line to the third line, we used the normality for length $L$, which holds for all fixed values of $s_3$.
}

As with MPS, we can construct an operator that reveals the virtual space of a normal SIMPS. Define the map, 
\begin{equation}
\begin{aligned}
&B_{[s_1,L,s_2]}(X)= \\ 
&\sum_{i_1\dots i_{L}}\mathrm{Tr}(B^{s_1i_1}B^{i_1i_2}\dots B^{i_{L}s_2} X)|i_1\dots i_{L}\rangle.
\end{aligned}
\end{equation}
If the SIMPS is normal and $L\geq L_1$, then $B_{[s_1,L,s_2]}$ has an inverse $B^{-1}_{[s_1,L,s_2]}$ such that $B^{-1}_{[s_1,L,s_2]}B_{[s_1,L,s_2]}=\id$. 
Using this, we can define the operator, \begin{equation} \label{eq:binv}
    B_L=\sum_{s_1,s_2} |s_1\rangle\langle s_1|\otimes B_{[s_1,L,s_2]}\otimes |s_2\rangle\langle s_2|,
\end{equation}
{For $L\geq L_1$, this map is invertible and its inverse $B_L^{-1}$ acts on $L+2$ spins and serves the purpose of revealing the virtual indices of the SIMPS. Later, we will use this map to construct string order parameters.}
One could also use this map to construct parent Hamiltonians for SIMPS which may have different properties than they corresponding MPS parent Hamiltonians, but we do not pursue this direction here.

In the next section, we will show that every normal MPS can be represented as a normal SIMPS and vice versa. Therefore, the nice physical properties of normal MPS such as the existence of gapped parent Hamiltonians and finite correlation length carry over directly to normal SIMPS.

\subsection{Converting between MPS and SIMPS} \label{sec:converting}

Now, we show how one can convert between an MPS and SIMPS representation of a state. In particular, we show that the SIMPS representation differs fundamentally  from the MPS representation only when some of the matrices $A^i$ are not full-rank for some physical basis $\{i \}$. 

First, suppose we have a SIMPS tensor $B^{ij}$. We can convert the SIMPS to an MPS by cutting $B^{ij}$ in half via singular value decomposition (SVD). Namely, define $\mathbb{B}=\sum_{i,j,\alpha,\beta} B^{ij}_{\alpha\beta} |i\alpha\rangle\langle j\beta|$ {which is a $\chi\times\chi$ matrix with $\chi=\sum_i \chi_i$} and use the SVD to write $\mathbb{B}=usv^\dagger$ where $u,v$ are isometries mapping from $\mathbb{C}^\chi$ to $\mathbb{C}^D$ satisfying $u^\dagger u = v^\dagger v=\id$ and $s$ is a $D\times D$ diagonal matrix. Therein, $D$ is the rank of $\mathbb{B}$. For compactness, let us absorb $s$ into $v$, defining $w^\dagger=sv^\dagger$.
From this, we define the matrices $A^i$ as,
\begin{equation} \label{eq:simpstomps}
    A^i = w^\dagger(|i\rangle\langle i|\otimes \id)u\ .
\end{equation}
Then, it can be checked that Eq.~\eqref{eq:simps} can be equivalently written as an MPS \eqref{eq:mps} with bond dimension $D$ defined by the matrices $A^i$.

Now, we show how to convert an MPS to a SIMPS. Naively, every MPS can be written as a SIMPS by simply ``discarding'' one of the physical indices: Given an MPS tensor $A^i$, we can define a SIMPS tensor that generates the same state as $B^{ij} = A^i$. However, this will not lead to a SIMPS with minimal bond dimension in general. Instead, we first write $A^i=A^iP^iP^{i\dagger}$ where $P^iP^{i\dagger}$ is the projector onto the domain of $A^i$, and $P^i:\mathbb{C}^{\chi^i} \rightarrow \mathbb{C}^D$ is an isometry where $\chi^i$ is the rank of $A^i$ such that $P^{i\dagger}P^i=\id_{\chi^i}$. Then, we define the SIMPS tensor,
\begin{equation} \label{eq:mpstosimps}
    B^{ij}=P^{i\dagger}A^jP^j\ .
\end{equation}
One can straightforwardly confirm that the SIMPS generated by the tensor $B^{ij}$ defined in this way is the same state as the MPS generated by $A^i$. The matrices $B^{ij}$ have dimension $\chi^i\times \chi^j$.

In the generic case, $A^i$ is full rank so $P^i=\id$ and $B^{ij}=A^i$. The cases where the SIMPS representation is distinct from the MPS representation are therefore those for which $A^i$  is not full-rank in some physical basis $\{i\}$. This leads us to the following definition of when a MPS ``makes a good SIMPS''; \text{i.e.}, when the SIMPS and MPS representations differ,
\begin{Def}[Good SIMPS]
    A normal MPS $A^i$ makes a good SIMPS if the matrix $A^i$ is rank deficient for some $i$.
\end{Def}
Since generic matrices have full rank, it may seem that only very fine-tuned MPS make good SIMPS. However, we will see that certain physical conditions, such as enforcing non-onsite symmetries or local constraints, will imply that an MPS makes a good SIMPS (see Props. \ref{theorem:nononsite_good} and \ref{theorem:rydberg_good}).

The conversions described above are optimal, in the sense that they lead to MPS and SIMPS of minimal bond dimension. This is a consequence of the following result which is proven in Appendix \ref{app:conversion},
\begin{Prop}[Optimal conversion] \label{prop:conversion}
    Given a SIMPS tensor $B^{ij}$ that is normal with injectivity length $L_1$, the MPS tensor defined in Eq.~\eqref{eq:simpstomps} is normal with injectivity length $L_0\leq L_1+2$. Similarly, given an MPS tensor $A^i$ that is normal with injectivity length $L_0$, the SIMPS tensor defined in Eq.~\eqref{eq:mpstosimps} is normal with injectivity length $L_1\leq L_0$.
\end{Prop}

\subsection{Fundamental theorem of SIMPS}

We now state the SIMPS equivalent of the fundamental theorem of MPS. The theorem describes how two normal SIMPS tensors generating the same state are related. For MPS, the fundamental theorem states that if $A^i$ and $B^i$ are both normal and generate the same MPS, then there is an invertible matrix $V$ such that $A^i=VB^i V^{-1}$, \text{i.e.}, the two tensors are related by a so-called gauge transformation \cite{Perez-Garcia2008}. For injective SIMPS, gauge transformations like $B^{ij} \rightarrow VB^{ij} V^{-1}$ can still be used to obtain new tensors generating the same state. However, it turns out that a larger class of transformations is also possible,
\begin{Theorem}[Fundamental Theorem of SIMPS] \label{theorem:fundamental}
    Let $A^{ij}$ and $B^{ij}$ be two normal SIMPS tensors that generate the same state for all system sizes. Then, there is a set of invertible matrices $V_i$ for $i=0,\dots,d-1$ such that,
    \begin{equation} \label{eq:fundamental}
        A^{ij} = V_iB^{ij}V_j^{-1}.
    \end{equation}
\end{Theorem}
\noindent
This is proven in Appendix \ref{app:fundamental}. The fundamental theorem shows that, once the basis of the physical spins is chosen, the bond dimension of all normal SIMPS representations of a state is the same. Therefore, it is well-defined to refer to the SIMPS bond dimension of a state for a given physical basis. The fundamental theorem will be essential to our understanding on non-onsite symmetries in SIMPS.

\subsection{Examples of SIMPS}

Here we give several examples of states which highlight differences between the SIMPS and MPS representations. More examples will appear in later sections of the paper.

\vspace{3mm}
\noindent
\textit{Example 1: The cluster state.} First, we consider the 1D cluster state, which is a paradigmatic resource state for MBQC \cite{Raussendorf2001a}, and also a simple example of a state with SPT order \cite{Son2012}. The cluster state can be written as,
\begin{equation} \label{eq:cluster}
    |C\rangle= \prod_{i=1}^N CZ_{i,i+1} |++\cdots +\rangle,
\end{equation} 
where $|\pm\rangle = \frac{1}{\sqrt{2}}(|0\rangle \pm |1\rangle)$ and $CZ$ is the controlled-$Z$ gate defined by $CZ|ij\rangle = (-1)^{ij}|ij\rangle$.
From this circuit representation, one can obtain the following MPS representation of the cluster state,
\begin{equation} \label{eq:cluster_mps}
        C^0=|0\rangle\langle +|, \quad  C^1=|1\rangle\langle -|\ .
\end{equation}
Observe that both $C^0$ and $C^1$ are rank one. Therefore, the cluster state can be written as a SIMPS with bond-dimension $\chi=1$ where the SIMPS tensor is defined as,
\begin{equation}
    B^{ij}=(-1)^{ij}, 
\end{equation}
where $i,j\in\{0,1\}$. If we instead consider the physical $X$-basis spanned by the states $|\pm\rangle$, then the MPS reads,
\begin{equation} \label{eq:cluster_mps_X}
    \begin{aligned}
        &C^+ := \frac{1}{\sqrt{2}}(C^0 + C^1)= H\ ,\\
        &C^- := \frac{1}{\sqrt{2}}(C^0 - C^1)= HZ\ ,
    \end{aligned}
\end{equation}
where $H=\frac{X+Z}{\sqrt{2}}$ is the Hadamard matrix. These matrices have full rank, so the SIMPS representation in the physical $X$-basis will not differ from the MPS representation. This showcases the importance of the physical basis in the SIMPS representation.

\vspace{3mm}
\noindent
\textit{Example 2: The nice SIMPS.} 
This example will serve to demonstrate many of the advantageous properties of SIMPS. In Sec.~\ref{sec:sptnononsite}, we will show that the state has SPT order protected by a non-onsite symmetry. The state can be defined by the following MPS tensor with $d=2$ and $D=3$,
\begin{equation} \label{eq:mps_nice_ex}
    A^0 =\begin{pmatrix}
        1 & 0 & 0 \\
        0 & 0 & 0 \\
        0 & 1 & 1
    \end{pmatrix}
    ,\quad
    A^1 =\begin{pmatrix}
        0 & 1 & -1 \\
        -1 & 0 & 0 \\
        0 & 0 & 0
    \end{pmatrix}.
\end{equation}
It is straightforward to verify that this MPS is normal, so it describes the unique ground state of a gapped Hamiltonian. Both matrices are rank deficient, having rank equal to two. Therefore, using Eq.~\eqref{eq:mpstosimps}, we can derive the corresponding SIMPS tensor with bond dimension $\chi=2$,
\begin{equation} \label{eq:simps_nice_ex}
    \begin{array}{lll}
        B^{00}=\id & B^{01} = \id   \\[\medskipamount]
        B^{10}= X & B^{11} = XZ 
    \end{array}
\end{equation}
Notice that the matrices $B^{ij}$ are all unitary, and are in fact all Pauli operators. This will be important for understanding the use of this state as a resource for quantum teleportation in Sec.~\ref{sec:simps_wire}. We remark that MPS with similar structure to the one above first appeared in Ref.~\cite{Wahl2012}, as discussed further in Sec.~\ref{sec:wahl}.

\vspace{3mm}
\noindent
\textit{Example 3: GHZ state.} As a simple example of a state with long-range order, consider the Greenberger–Horne–Zeilinger (GHZ) state,
\begin{equation}
    |GHZ\rangle = \frac{1}{\sqrt{2}}(|00\cdots 0\rangle + |11\cdots 1\rangle)\ ,
\end{equation}
which has the $D=2$ MPS representation $A^i = |i\rangle\langle i|$. The cat-like entanglement and long-range order of this state shows that this is not a normal MPS. Since the matrices $A^i$ have rank one, the GHZ state can be written as a non-normal SIMPS with bond-dimension $\chi=1$ by writing $B^{ij}=\delta_{ij}$.

\vspace{3mm}
\noindent
\textit{Example 3: State with anomalous symmetry.} For a more interesting example of a state with long-range order, consider the following state from Ref.~\cite{GarreRubio2023},
\begin{equation}
    |\psi\rangle = \frac{1}{\sqrt{2}}(|++\cdots +\rangle + \prod_i Z_iCZ_{i,i+1}|++\cdots +\rangle)\ .
\end{equation}
This state is invariant under the operator $U_{CZX}=\prod_i X_i \prod_{i=1}^N Z_iCZ_{i,i+1}$ which is an anomalous $Z_2$ symmetry \cite{Chen2011,GarreRubio2023}. An MPS representation for the state can be obtained as a direct sum of the MPS representations of $|++\cdots +\rangle$ and  $\prod_i Z_iCZ_{i,i+1}|++\cdots +\rangle)$,
\begin{equation} \label{eq:mps_czb}
    A^0 =\begin{pmatrix}
        1 & 0 & 0 \\
        0 & 1 & 1 \\
        0 & 0 & 0
    \end{pmatrix}
    ,\quad
    A^1 =\begin{pmatrix}
        1 & 0 & 0 \\
        0 & 0 & 0 \\
        0 & -1 & 1
    \end{pmatrix}.
\end{equation}
These matrices each have rank two. The corresponding SIMPS tensor then has $\chi=2$,
\begin{equation} \label{eq:simps_anom_ex}
    \begin{array}{lll}
        B^{00}=\id & B^{01} = \id   \\[\medskipamount]
        B^{10}=Z & B^{11} = \id \ .
    \end{array}
\end{equation}
Note that, because the matrices $B^{ij}$ are all diagonal, this SIMPS is not normal.

\section{Non-onsite symmetries and SPT phases in SIMPS} \label{sec:spt}

In this section, we demonstrate the one context where SIMPS are advantageous over MPS: capturing non-onsite symmetries. First, we review how onsite symmetries are captured in MPS. Then, we show how similar ideas can be used to analyze certain non-onsite symmetries in SIMPS. This includes a general class of anomalous group-like symmetries. We then define a class of SIMPS possessing non-trivial SPT order protected by non-onsite symmetries and characterize their physical properties.

\subsection{SPT order in MPS} \label{sec:mps_spt}

Here we review how MPS can represent symmetries and SPT order. Suppose an MPS has a global onsite symmetry on PBC of the form $u^{\otimes N} |\psi\rangle = |\psi\rangle$. Then we have the relation \cite{Perez-Garcia2008},
\begin{equation} \label{eq:mps_symm}
\tilde{A}^i := \sum_j u_{ij} A^j = V A^i V^\dagger,
\end{equation}
for some unitary matrices $V$, where $u_{ij}=\langle i|u|j\rangle$. To see how this relation implies the global symmetry, we can write,
\begin{equation} \label{eq:symm_pushthrough}
    \begin{aligned}
        &u^{\otimes N}|\psi\rangle = \sum_{i_1,\dots, i_N} \mathrm{Tr}(\tilde{A}^{i_1}\cdots \tilde{A}^{i_N})|i_1\cdots i_N\rangle \\
        &= \sum_{i_1,\dots, i_N} \mathrm{Tr}(VA^{i_1}V^\dagger\cdots VA^{i_N}V^\dagger)|i_1\cdots i_N\rangle \\
        &= \sum_{i_1,\dots, i_N} \mathrm{Tr}(A^{i_1}\cdots A^{i_N})|i_1\cdots i_N\rangle = |\psi\rangle
    \end{aligned}
\end{equation}
where the operators $V$ and $V^\dagger$ cancelled pairwise in the third step. Therefore, MPS are able to encode global symmetries in terms of local symmetries of the single tensors.

If there is a set of symmetries $u_g$ that form a representation of a group $G$ such that $u_gu_h = u_{gh}$, then the corresponding operators $V_g$ that appear in Eq.~\eqref{eq:mps_symm} form a projective representation, $V_gV_h = \omega(g,h) V_{gh}$ for some phase factors $\omega(g,h)$ \cite{Schuch2011}. If $G$ is finite abelian, which we always assume from now on, we can instead characterize the representation by the relation $V_gV_h = \Omega(g,h) V_h V_g$ where $\Omega(g,h)=\omega(g,h)\omega(h,g)^{-1}$. In the case of finite abelian $G$, the phases $\Omega(g,h)$ uniquely label the SPT phase to which the MPS belongs \cite{Berkovich1998}. If these phases are not all equal to one, then the MPS belongs to a non-trivial SPT phase protected by $G$ \cite{Pollmann2010,Chen2011,Schuch2011}.

If Eq.~\eqref{eq:mps_symm} is satisfied, it is straightforward to show that the MPS on OBC \eqref{eq:mps_obc} is invariant under the operators $V_g\otimes u_g^{\otimes N}\otimes V_g^*$ where $V_g$ and $V_g^*$ act on the boundary spins where $*$ denotes complex conjugation. The fact that the boundary spins in non-trivial SPT phases carry a projective representation implies a symmetry protected degeneracy at the physical boundary of the chain and in the entanglement spectrum \cite{Pollmann2010}.

\vspace{3mm}
\noindent
\textit{Example: the cluster state.} The cluster state is symmetric under a $Z_2\times Z_2$ symmetry generated by $X_A = \prod_{i=1}^{N/2} X_{2i}$ and $X_B = \prod_{i=1}^{N/2} X_{2i+1}$ \cite{Son2012}. Since this symmetry has a translation unit cell of length two, we must use a two-spin unit cell to express the symmetry in an onsite way. In a two-spin unit cell, the MPS tensor for the cluster state can be defined as,
\begin{equation} \label{eq:cluster_mps_X_2site}
    \begin{aligned}
        &C_2^{++} := C^+C^+ = \id\\
        &C_2^{+-} := C^+C^- = Z\\
        &C_2^{-+} := C^-C^+ = X\\
        &C_2^{--} := C^-C^- = XZ\ .
    \end{aligned}
\end{equation}
From this, Eq.~\eqref{eq:mps_symm} looks like,
\begin{equation} \label{eq:cluster_mps_symmetries}
    \begin{aligned}
        &\sum_{ij}[XI]_{ij,i'j'}C_2^{i'j'} = Z C_2^{ij}Z\\
        &\sum_{ij}[IX]_{ij,i'j'}C_2^{i'j'} = X C_2^{ij}X\ .
    \end{aligned}
\end{equation}
Since $X$ and $Z$ anticommute, this shows that the cluster state has nontrivial SPT order under the $Z_2\times Z_2$ symmetry. Physically, this anticommutation manifests as a two-fold symmetry protected degeneracy in the edge spectrum and the entanglement spectrum.

\subsection{Non-onsite symmetries in SIMPS} \label{sec:simps_symm}

Now, we show that SIMPS are better suited to capturing a certain class of non-onsite symmetries compared to MPS. Let $U$ be a two-body unitary gate that is diagonal in the physical basis of the SIMPS, such that $U|ij\rangle = e^{i\varphi_{ij}}|ij\rangle$ for some phases $e^{i\varphi_{ij}}$. Define the symmetry operator $\mathcal{U}=\prod_{i=1}^N U_{i,i+1}$. For generic choices of $\varphi_{ij}$, this operator is fundamentally \textit{non-onsite}: it cannot be written in the form $u^{\otimes N}$ for some unitary $u$, even after arbitrary enlargement of the unit cell. The operator $\mathcal{U}$ acts on a SIMPS as follows,
\begin{equation}
    \mathcal{U}|\psi\rangle = \sum_{i_1,\dots, i_N} \mathrm{Tr}(\tilde{B}^{i_1i_2}\tilde{B}^{i_2i_3}\cdots \tilde{B}^{i_Ni_1}) |i_1i_2\cdots i_N\rangle.
\end{equation}
where we have defined $\tilde{B}^{ij} = e^{i\varphi_{ij}}B^{ij}$. This equation is non-trivial: It shows that the non-onsite symmetry $\mathcal{U}$ acts on a SIMPS by transforming each tensor without changing the bond dimension. This is possible since the split-index structure effectively splits the different gates $U$ across different tensors.
Because of this, if $\mathcal{U}|\psi\rangle=|\psi\rangle$ for all system sizes and the SIMPS is normal, then Theorem \ref{theorem:fundamental} says that there must exist invertible matrices $V_i$ such that,
\begin{equation} \label{eq:qc_fund_thm}
    e^{i\varphi_{ij}}B^{ij} = V_iB^{ij}V_{j}^{-1}\ .
\end{equation}
Note that this was possible because we assumed that $\mathcal{U}$ is diagonal. More general symmetries can also correspond to SIMPS gauge transformations as long as the off-diagonal parts are a product of onsite permutation operators (\textit{e.g.} $\prod_i X_i$) since such operators can also be split to act on each tensor independently.  
Therefore, we can understand certain non-onsite symmetries of SIMPS in terms of simple gauge transformations. We discuss this class of symmetries in more detail in Sec.~\ref{sec:anomalous}.

Non-onsite symmetries can also be understood with MPS \cite{Chen2011a}, but the treatment is more complicated. Since acting with a non-onsite symmetry increases the bond dimension of an MPS, the symmetry cannot be understood in terms of simple gauge transformations, and instead one needs to consider special three-index reduction tensors that reduce the bond dimension and satisfy certain pulling-through conditions \cite{Chen2011a,GarreRubio2023}. Thus, SIMPS significantly simplify the treatment of certain non-onsite symmetries compared to MPS. 

Using Eq.~\eqref{eq:qc_fund_thm}, we can prove that every MPS which is invariant under a diagonal non-onsite symmetry makes a good SIMPS,
\begin{Prop} \label{theorem:nononsite_good}
    Suppose a normal MPS $A^i$ is invariant under a symmetry of the form $\mathcal{U}=\prod_{i=1}^N U_{i,i+1}$ where $U_{i,i+1}$ is diagonal and is not a product of single-site operators. Then, at least one of the matrices $A^i$ must be not full rank.
\end{Prop}

\begin{proof}
    We prove the contrapositive. Assume that every matrix $A^i$ has full rank. Then, according to Eq.~\eqref{eq:mpstosimps}, the corresponding SIMPS tensor can be chosen as $B^{ij}=A^j$. Now, suppose that the MPS is invariant under $\mathcal{U}$ such that Eq.~\eqref{eq:qc_fund_thm} holds,
\begin{equation}
    e^{i\varphi_{ij}}A^j = V_iA^jV_j^{-1}.
\end{equation}
Using the fact that $A^i$ are all invertible, we can write $e^{i\varphi_{ij}} = V_iA^jV_j^{-1}{A^{j}}^{-1}$, which gives,
\begin{equation}
    e^{i\varphi_{i j}}e^{-i\varphi_{i' j}} = V_{i}V_{i'}^{-1}.
\end{equation}
Since the right-hand side is independent of $j$, the combination of phases on the left-hand side must also be independent of $j$. Then, it we write $e^{i\varphi_{ij}}=(e^{i\varphi_{ij}}e^{-i\varphi_{0j}})(e^{i\varphi_{0j}})$, we see that the phase in the first parentheses depends only on $i$ and the phase while the second parentheses only depends on $j$. Therefore, the two-site operator $U$ with $U|ij\rangle = e^{i\varphi_{ij}}|ij\rangle$ is a product of single-site operators.
\end{proof}

This result can be used to rule out non-onsite symmetries of MPS. For example, since the cluster state MPS has full-rank in the $X$-basis, we know that there are no non-onsite symmetries which are diagonal in this basis.

Let us study non-onsite symmetries using the ``nice SIMPS'' example defined in Eq.~\eqref{eq:simps_nice_ex}. By inspection of Eq.~\eqref{eq:simps_nice_ex}, we can immediately derive the following symmetries of the SIMPS tensor,
\begin{equation} \label{eq:nice_symmetry}
    \begin{aligned}
        &XB^{ij}X = (-1)^{ij} B^{ij} \\
        &ZB^{ij}Z = (-1)^i B^{ij}
    \end{aligned}
\end{equation}
These symmetries of the SIMPS tensor imply global symmetries of the SIMPS. Namely, we find that,
\begin{equation}
    \overline{CZ}:=\prod_{i=1}^N CZ_{i,i+1}\quad \text{and} \quad \overline{Z}:=\prod_{i=1}^N Z_i
\end{equation} 
are both symmetries of the SIMPS. The former is a non-onsite symmetry, while the latter is onsite, and together they form a $Z_2\times Z_2$ symmetry group. Later, we will define a large class of SIMPS that have a pair of symmetries that are generically both non-onsite.

In this example, the non-onsite symmetries could be identified immediately once the SIMPS representation of the states is written down. Conversely, it is not at all obvious to see these symmetries from the MPS representation in Eq.~\eqref{eq:mps_nice_ex}. This is another advantageous property of SIMPS when it comes to analyzing non-onsite symmetries.

The relations in Eq.~\eqref{eq:nice_symmetry} suggest that the SIMPS has non-trivial SPT order under the $Z_2\times Z_2$ symmetry group. While the symmetry operators commute, their representations in the virtual space of the SIMPS anticommute, and this was used to demonstrate non-trivial SPT order for the cluster state in Eq.~\eqref{eq:cluster_mps_symmetries}. Demonstrating and characterizing this SPT order is the goal of Sec.~\ref{sec:sptnononsite}.

\subsection{Anomalous symmetries} \label{sec:anomalous}

Non-onsite symmetries are often associated with anomalies, and here we show that SIMPS are well-suited to capturing anomalous symmetries. In the present context, an anomalous symmetry can be described as a symmetry which does not permit a short-range entangled symmetric state \cite{Chen2011a}. Given a group of non-onsite symmetries, one can calculate an invariant called a 3-cocycle \cite{Chen2011a,Else2014}. If this 3-cocycle belongs to a non-trivial equivalence class, then it can be shown that no injective MPS can be symmetric under this group \cite{Chen2011a}. The alternatives are (a) the MPS is non-injective, meaning it has long-range order \text{i.e.}, symmetry breaking, or (b) the state is not representable as an MPS with finite bond dimension, such as critical states of gapless Hamiltonians. 


To demonstrate the ability of SIMPS to capture anomalous symmetries, we now write down a representative set of non-onsite symmetries for every finite group $G$ and 3-cocycle $\omega$ such that the symmetries can be captured by SIMPS gauge transformations. 
The representations, described in Ref.~\cite{GarreRubio2023}, are defined with respect to an onsite Hilbert space $\mathbb{C}^{|G|}$ spanned by the states $|g\rangle$ for $g\in G$. Define the operators $L^g$ such that $L^g|h\rangle = |gh\rangle$ and the diagonal two-site operators $W^g$ such that $W^g|hk\rangle = \omega(g,k,k^{-1}h)|hk\rangle$. Then, the non-onsite symmetries are defined for each $g\in G$ as,
\begin{equation}
    T^g = \prod_i L^g_i \prod_i W^g_{i,i+1}\ .
\end{equation}
One can check that these non-onsite operators, as a representation of $G$, are characterized by the 3-cocycle $\omega$ \cite{GarreRubio2023}.
This operator consists of a product of diagonal two-site operators and single-site permutations, and therefore it can be captured via SIMPS gauge transformations. Specifically, acting with this symmetry on a SIMPS transforms the tensor in the following way,
\begin{equation}
    B^{h,k}\xrightarrow[]{T^g} \omega(g,k,k^{-1}h) B^{gh,gk}\ .
\end{equation}
If a normal SIMPS is invariant under this symmetry, the fundamental theorem \eqref{eq:fundamental} says that there exist invertible matrices $V^g_h$ such that,
\begin{equation}
    \omega(g,k,k^{-1}h) B^{gh,gk} = V^g_h B^{h,k} {V^g_k}^{-1}\ .
\end{equation}
Strictly speaking, Ref.~\cite{Chen2011a} proved that no normal MPS (and therefore no normal SIMPS), can be invariant under all $T^g$ when $\omega$ is non-trivial. However, we expect that the fundamental theorem of SIMPS can be extended to non-normal SIMPS, as is the case for MPS (see Theorem IV.4 of Ref.~\cite{Cirac2021}, for example).
This is the case for the SIMPS in Eq.~\eqref{eq:simps_anom_ex} which is symmetric under the anomalous $Z_2$ symmetry $U_{\mathit{CZX}}$ and is accordingly not normal. Nevertheless, the symmetry acts on the SIMPS tensor as,
\begin{equation}
    B^{i,j}\xrightarrow[]{U_{CZX}} (-1)^{i+ij}B^{i+1,j+1}\ .
\end{equation}
which is indeed equivalent to a gauge transformation for the tensor defined in Eq.~\eqref{eq:simps_anom_ex},
\begin{equation}
    (-1)^{i+ij}B^{i+1,j+1} = (Z^i X) B^{i,j}(Z^j X)^{-1}\ .
\end{equation}
Therefore, we expect that SIMPS will be a useful tool for studying anomalous symmetries in future.

\subsection{SPT order with non-onsite symmetries} \label{sec:sptnononsite}

In this section, we construct a large family of SIMPS that serve as fixed-point states for SPT phases with (non-anomalous) non-onsite symmetries. We then characterize these SIMPS in terms of string order parameters, response to flux insertion, entanglement spectrum signatures, and by characterizing the universal patterns of entanglement which we call ``universal fingerprints''.

Define the following family of SIMPS with onsite dimension $d$ and bond dimension $\chi=2$,
\begin{equation} \label{eq:fp_simps}
B_\ab^{ij}=X^{a_{ij}}Z^{b_{ij}},
\end{equation}
for some $d\times d$ binary matrices $\bf{a}$, $\bf{b}$. We call states generated by these SIMPS tensors $|\psi_\ab\rangle$. This form includes the examples in Eq.~\eqref{eq:simps_nice_ex} and \eqref{eq:simps_anom_ex}. More generally, we could also consider $\chi>2$ by replacing $X$ and $Z$ with generalized Pauli operators, but we focus on the case of $\chi=2$ for simplicity. The matrices $B_\ab^{ij}$ satisfy the relations,
\begin{equation} \label{eq:fp_simps_symms_1}
    \begin{aligned}
        &(-1)^{b_{ij}}B_\ab^{ij} = XB_\ab^{ij}X \\
        &(-1)^{a_{ij}}B_\ab^{ij} = ZB_\ab^{ij}Z\ .
    \end{aligned}
\end{equation}
These relations imply non-onsite symmetries of $|\psi_\ab\rangle$ of the form,
\begin{equation} \label{eq:ab_symms}
    U^{\bf{a}} := \prod_i u^{\bf{a}}_{i,i+1}\quad\text{and}\quad U^{\bf{b}} := \prod_i u^{\bf{b}}_{i,i+1}\ ,
\end{equation} 
where $u^{\bf{a}}|ij\rangle = (-1)^{a_{ij}}|ij\rangle$ and $u^{\bf{b}}|ij\rangle = (-1)^{b_{ij}}|ij\rangle$. These symmetries form a $Z_2\times Z_2$ symmetry group which we call $G_\ab$. If $\bf{a}$ cannot be factorized, meaning $a_{ij}\neq c_i d_j$ for some binary vectors $\bf{c},\bf{d}$, then $U^{\bf{a}}$ is a non-onsite symmetry (and similar for $\bf{b}$). Eq.~\eqref{eq:fp_simps_symms_1} shows that the two generators anticommute in the virtual space of the SIMPS, so the states $|\psi_\ab\rangle$ will be representatives of non-trivial SPT phases protected by $G_\ab$. In the rest of this section, we only consider choices of $\ab$ such that $B_\ab^{ij}$ define normal SIMPS, which includes Eq.~\eqref{eq:simps_nice_ex} but not Eq.~\eqref{eq:simps_anom_ex}. One can show that this implies that $U^{\bf{a}}$ and $U^{\bf{b}}$ form a faithful representation of $Z_2\times Z_2$ for sufficiently large chain lengths (see Appendix \ref{app:fingerprints}).

\subsubsection{String order parameters}

The states $|\psi_\ab\rangle$ all possess non-trivial string order. To see this, observe that Eq.~\eqref{eq:fp_simps_symms_1} implies that applying $u^{\bf{a}}$ ($u^{\bf{b}}$) to a finite segment of the chain is equivalent to inserting virtual $Z$ ($X$) operators in the virtual space of the SIMPS and the endpoints of that segment. Since we assume that the SIMPS is normal, it corresponds to a normal MPS. Then, one can always find an operator acting on $L_0$ physical spins that has the same effect on the SIMPS as inserting these virtual Paulis. One could also derive this physical operator by using the map $B_L$ defined in Eq.~\eqref{eq:binv}.
Doing this, we can define the string operators,
\begin{equation} \label{eq:sop}
    S^{{\bf{a}/\bf{b}}}(\ell,r) = O^{\bf{a}/\bf{b}}_\ell \left(\prod_{i=\ell}^{r-1} u^{\bf{a}/\bf{b}}_{i,i+1}\right) O^{\bf{a}/\bf{b}}_r
\end{equation}
where $\ell$ and $r$ specify the locations of the left and right ends of the string and $O^{\bf{a}}_{\ell/r}$ ($O^{\bf{b}}_{\ell/r}$) are the aforementioned operators that act in a finite region around sites $\ell/r$ and have the effect of inserting $Z$ ($X$) in the virtual space of the SIMPS. By construction, we have $\langle \psi_\ab|S^{{\bf{a}/\bf{b}}}(\ell,r)|\psi_\ab\rangle=1$, \text{i.e.}, perfect string order.

The operators $O^{\bf{a}}_{\ell/r}$ ($O^{\bf{b}}_{\ell/r}$) carry a $-1$ charge under $U^{\bf{b}}$ ($U^{\bf{a}}$), which demonstrates the non-triviality of the string order. Because of this, we can use the same argument as in Ref.~\cite{Huang2015} to show that no finite-depth quantum circuit whose gates commute with $G_\ab$ can map $|\psi_\ab\rangle$ to a symmetric product state.\footnote{Since the symmetry group $G_\ab$ is not anomalous, there do exist symmetric product states such as $|00\cdots 0\rangle$.} This is a direct proof of the non-trivial SPT order of these states.

Finally, pushing the string order parameters to the ends of a state with open boundary conditions shows that the edges transform under a projective representation of $Z_2\times Z_2$, which demonstrates that the edge spectrum must be twofold degenerate, as in conventional SPT phases with onsite symmetries.

\subsubsection{Entanglement spectrum}

An important signature of 1D SPT order with onsite symmetries is half-chain entanglement spectrum degeneracy \cite{Pollmann2010}. We can divide a chain of length $N$ into two parts, $A$ and $B$, containing sites $1,\dots, j$ and $j+1,\dots,N$ respectively. The entanglement spectrum is then defined as the (logarithm of the) spectrum of $\rho_A = \mathrm{Tr}_B|\psi\rangle\langle \psi|$. For states possessing 1D SPT order with onsite symmetries, there will be an exact degeneracy in this spectrum, where the magnitude of the degeneracy depends on the particular symmetry and SPT phase \cite{Pollmann2010}. For example, for all states in the same SPT phase as the 1D cluster state, the eigenvalues of $\rho_A$ are all at least twofold degenerate. This is a manifestation of the bulk-boundary correspondence \cite{Li2008}, as the entanglement spectrum degeneracy reflects the boundary degeneracy of SPT phases.

For SPT phases with non-onsite symmetries, entanglement spectrum degeneracy is no longer present in general. Indeed, the example in Eq.~\eqref{eq:simps_nice_ex} has MPS bond dimension $D=3$, meaning its entanglement spectrum has three non-zero eigenvalues, so it cannot possibly have a two-fold degenerate entanglement spectrum. Indeed, the entanglement spectrum is $(1/2,1/4,1/4)$, and the degeneracy of second and third eigenvalues is not protected by symmetry. This lack of degeneracy can be attributed to the non-onsite nature of the symmetries. Roughly speaking, when we ``cut'' the chain to compute the entanglement spectrum, we also need to cut the symmetries themselves (which is not needed for onsite symmetries) and this ruins the bulk-boundary correspondence.

There is, however, a hidden entanglement spectrum degeneracy which can be revealed using the SIMPS representation. Consider a SIMPS with OBC \eqref{eq:simps_obc} and project the spin at site $k$ onto the state $|s\rangle$ (by, \text{e.g.}, measuring it). This results in the state $|s\rangle_k \otimes |\psi_s\rangle$ where,
\begin{equation*} 
\begin{aligned} \label{eq:simps_s}
    |\psi^s\rangle = \sum_{\substack{i_1,\dots, i_N \\ \alpha,\beta}} &\langle \alpha |B^{i_1i_2}\cdots B^{i_{k-1}s}B^{s i_{k+1}}\cdots B^{i_{N-1}i_N}|\beta \rangle \\ &\times |\alpha,i_1i_2\cdots i_{k-1}i_{k+1}\cdots i_N,\beta\rangle\ ,
\end{aligned}
\end{equation*}
where the sum is over all indices except $i_k$ which has been fixed to $i_k=s$. Cutting the state between sites $k-1$ and $k+1$ gives,
\begin{equation}
    |\psi^s\rangle = \sum_{\gamma=1}^\chi |L^s_\gamma\rangle |R^s_\gamma\rangle\ ,
\end{equation}
where,
\begin{equation*}
\begin{aligned}
    L^s_\gamma &= \sum_{\substack{i_1,\dots, i_{k-1} \\ \alpha}} \langle \alpha|B^{i_1i_2}\cdots B^{i_{k-1}s}|\gamma\rangle |\alpha,i_1\cdots i_{k-1}\rangle  \\
    R^s_\gamma &= \sum_{\substack{i_1,\dots, i_{k-1} \\ \beta}} \langle \gamma|B^{si_{k+1}}\cdots B^{i_{N-1}i_{N}}|\beta\rangle |i_{k+1}\cdots i_N,\beta\rangle
\end{aligned}
\end{equation*}
The indices $\alpha,\beta,\gamma$ all run from $0$ to $\chi-1$, so there are at most $\chi$ non-zero eigenvalues in the entanglement spectrum when we cut between sites $k-1$ and $k+1$. Furthermore, since the SIMPS virtual space (where the index $\gamma$ lives) transforms under a projective representation of $Z_2\times Z_2$ \eqref{eq:fp_simps_symms_1}, one can argue in the standard way that the spectrum must be twofold degenerate. This works because fixing site $k$ in the state $|s\rangle$ effectively disentangles the non-onsite symmetry, such that the symmetry factorizes between the left and right halves of $|\psi^s\rangle$, and the entanglement spectrum degeneracy is recovered.

\subsubsection{Symmetry twists and flux insertion}
\label{sec:twists}

SIMPS also give a very simple way to introduce symmetry twists corresponding to non-onsite symmetries.
In an MPS, a symmetry twist corresponding to the onsite symmetry $u^{\otimes N}$ is defined by inserting a symmetry flux $V$ into the MPS virtual space,
\begin{equation}
    |\tilde{\psi}\rangle = \sum_{i_1,\dots, i_N} \mathrm{Tr}(V A^{i_1}A^{i_2}\cdots A^{i_N}) |i_1i_2\cdots i_N\rangle
\end{equation}
where $V$ is the virtual counterpart of $u$ as defined in Eq.~\eqref{eq:mps_symm}. For non-onsite symmetries, there is no simple counterpart to $V$ to insert in an MPS, so it is not clear how to insert a symmetry twist. 

For SIMPS, however, we can easily introduce a twist corresponding to a non-onsite symmetry $\mathcal{U}$ by writing,
\begin{equation}
    |\tilde{\psi}\rangle = \sum_{i_1,\dots, i_N} \mathrm{Tr}(V^{i_1} B^{i_1i_2}B^{i_2i_3}\cdots B^{i_Ni_1}) |i_1i_2\cdots i_N\rangle.
\end{equation}
where $V^{i}$ is as defined in Eq.~\eqref{eq:qc_fund_thm}. Recall that, for the states $|\psi_\ab\rangle$, $V^i$ is independent of $i$. 

In SPT phases with an onsite abelian symmetry $G$, a symmetry twisted state with $V=V_g$ carries charge $\Omega(g,h)$ under global $h$ symmetry \cite{Shiozaki2017}. Since the charges of all twisted states determines $\Omega(g,h)$, this is an equivalent way to characterize SPT order.
The same thing works for the SIMPS corresponding to the states $|\psi_\ab\rangle$, where the state twisted by $U^{\bf{a}}$ carries a $-1$ charge under $U^{\bf{b}}$, and vice-versa.

\subsubsection{Universal fingerprints} \label{sec:fingerprints}

SPT phases are characterized by certain universal patterns of entanglement that cannot be removed by finite-depth circuits that respect the protecting symmetry \cite{Chen2013}. The framework of MPS can be used to precisely reveal these patterns, which we refer to as the \textit{universal fingerprints} of an SPT phase. Here, we review how this works for MPS and conventional SPT phases, and then we derive the corresponding result for non-onsite symmetries using SIMPS.

Consider an arbitrary state $|\phi\rangle$ in the same SPT phase as the cluster state $|C\rangle$, which we refer to as the cluster phase. It was shown in Ref.~\cite{Else2012} that $|\phi\rangle$ can be represented by an MPS of the following form,
\begin{equation} \label{eq:else}
    A_\phi^i = J_\phi^i\otimes C_2^i
\end{equation}
where $i=++,+-,-+,--$, $C_2$ is the cluster state tensor \eqref{eq:cluster_mps_X_2site}, and $J_\phi^i$ are arbitrary matrices encoding the microscopic details of $|\phi\rangle$. This equation implies a tensor product decomposition of the virtual space into the ``junk subsystem'' and ``logical subsystem'', which are acted on by $J_\phi$ and $C_2$, respectively. Note that this decomposition only holds for a single choice of basis of the physical spins, which is the basis that diagonalizes the symmetry operators. Alternative proofs of this fact can be found in Refs.~\cite{Lake2022, phdthesis}. This result shows that the ``fingerprint'' of the cluster state, as conveyed by the tensor $C_2^i$, is present in every state within the cluster phase. Universal fingerprint results are very powerful, having been used to understand hidden symmetry breaking \cite{Else2013} and the entanglement structure \cite{Stephen2019a,deGroot2020} in SPT phases. They are also the basis for schemes of MBQC using SPT phases as resources \cite{Else2012,Miller2015,Stephen2017,Raussendorf2019,Stephen2019,Daniel2019,Devakul2018mbqc} and have been used to construct exact renormalization circuits for SPT phases \cite{Lake2022}. Similar results have been obtained for other kinds of SPT phases, including 1D phases with spatially modulated symmetries \cite{Stephen2019} and 2D phases with subsystem symmetries \cite{Raussendorf2019, Stephen2019, Devakul2018mbqc, Daniel2019}

Now, we show how to uncover the universal fingerprints of SPT phases with non-onsite symmetries using SIMPS. 
Consider an arbitrary state belonging to the same SPT phase as $|\psi_\ab\rangle$. This phase is defined as the set of all states that can be related to $|\psi_\ab\rangle$ by a finite-depth quantum circuit consisting of gates that commute with $G_\ab$. In Appendix \ref{app:fingerprints}, we prove the following result,
\begin{Lemma} \label{lemma:fingerprints}
    Let $O$ be any operator supported on the interval $[x,y]$ that commutes with $G_\ab$. Then there exists another operator $O'$ supported on $[x,y]$ that is diagonal in the local $Z$-basis and satisfies $O|\psi_\ab\rangle=O'|\psi_\ab\rangle$. 
\end{Lemma}
\noindent
From this result, we have,
\begin{Theorem}[Universal Fingerprints] \label{theorem:fingerprints}
    Let $|\phi\rangle$ be a state in the same SPT phase as $|\psi_\ab\rangle$ \eqref{eq:fp_simps} with respect to $G_\ab$. Then $|\phi\rangle$ can be represented by a SIMPS defined by the matrices $B_\phi^{ij} = J_\phi^i\otimes B_\ab^{ij}$ for some matrices $J_\phi^i$, where $B_\ab^{ij}$ is the SIMPS representation of $|\psi_\ab\rangle$.
\end{Theorem}

\begin{proof}
    By definition, for any state $|\phi\rangle$ in the same SPT phase as $|\psi_\ab\rangle$, there exists a symmetric finite-depth quantum circuit $\mathcal{U}_\phi$ such that $|\phi\rangle =  \mathcal{U}_\phi|\psi_\ab\rangle$. Here, $\mathcal{U}_\phi = \prod_i u_i$ where each $u_i$ is a unitary gate with finite support such that the gate can be applied in a finite number of layers of non-overlapping gates, and each gate is symmetric, $[U^{\bf{a}},u_i]=[U^{\bf{b}},u_i] = 0$. 
    
    When then apply Lemma \ref{lemma:fingerprints} to the gates $u_i$. For the gates that act directly on $|\psi_\ab\rangle$, \text{i.e.}, the first layer, we can immediately apply Lemma \ref{lemma:fingerprints} to replace them with diagonal operators. For the gates in other layers, we need to first commute the gates through the lower layers before we can act on $|\psi_\ab\rangle$. This introduces some correlations between the different local operators, but these correlations remain local. By replacing each gate $u_i$ with a diagonal operator, we can replace the whole circuit $\mathcal{U}_\phi$ with a diagonal operator $\mathcal{J}_\phi$ such that $|\phi\rangle = \mathcal{J}_\phi|\psi_\ab\rangle$. This operator is no longer a finite-depth circuit in general (or even a unitary operator), but it can nonetheless be represented by a matrix product operator (MPO) of the form,
    \begin{equation}
        \mathcal{J}_\phi = \sum_{i_1,\dots i_N} \mathrm{Tr} (J_\phi^{i_1}\cdots J_\phi^{i_N})|i_1\cdots i_N\rangle\langle i_1\cdots i_N| \,
    \end{equation}
    where $J_\phi^i$ are some matrices of dimension $D_\phi$. Due to the local structure of the correlations in $\mathcal{J_\phi}$, the bond dimension $D_\phi$ is finite in the thermodynamic limit. Finally, we can obtain a SIMPS representation of $|\phi\rangle$ by applying this MPO to $|\psi_\ab\rangle$, resulting in,
    \begin{equation}
    \begin{aligned}
        &\langle i_1 \cdots i_N|\mathcal{J}_\phi|\psi_\ab\rangle\\ &=\mathrm{Tr}(J_\phi^{i_1}\cdots J_\phi^{i_N})\ \mathrm{Tr}(B_\ab^{i_1i_2}\cdots B_\ab^{i_Ni_{1}}) \\
        &=\mathrm{Tr}([J_\phi^{i_1}\otimes B_\ab^{i_1i_2}]\cdots [J_\phi^{i_N}\otimes B_\ab^{i_Ni_{1}}]),
    \end{aligned}
    \end{equation}
    such that we can represent $|\phi\rangle$ as a SIMPS defined by the matrices $B^{ij}_\phi = J^i_\phi\otimes B^{ij}_\ab$.
\end{proof}

Theorem \ref{theorem:fingerprints} implies that all of the interesting properties of the SIMPS observed in the previous sections are robust properties of the corresponding SPT phase. For example, the following relations hold for a SIMPS $B^{ij}_\phi$ in the same SPT phase as $B^{ij}_\ab$,
\begin{equation} \label{eq:gen_simps_symms_1}
    \begin{aligned}
        &(-1)^{b_{ij}}B_{\phi}^{ij} = (\id\otimes X)B_{\phi}^{ij}(\id\otimes X) \\
        &(-1)^{a_{ij}}B_{\phi}^{ij} = (\id\otimes Z)B_{\phi}^{ij}(\id\otimes Z)\ ,
    \end{aligned}
\end{equation}
so that the two generators of $G_\ab$ anticommute in the virtual space of the SIMPS for every state in the same SPT phase. This means that the previous conclusions about string order and (lack of) entanglement spectrum degeneracy hold throughout the entire SPT phase protected by the symmetry $G_\ab$.

We remark that the universal fingerprints can also be expressed using MPS. That is, if we let $A^i_{\ab}$ be the MPS tensor derived from $B^{ij}_{\ab}$, we can also express the state $|\phi\rangle = J_\phi |\psi_{\ab}\rangle$ as an MPS with tensor $A_\phi^i = J_\phi^i\otimes A^i_{\ab}$. Note, however, that the SIMPS representation was essential for proving Theorem \ref{theorem:fingerprints}.

\section{Quantum teleportation and computation with SIMPS} \label{sec:computation}

In this section, we move from the physical properties of SIMPS to their computational properties. We first review how MPS can be used as resources for long-range quantum teleportation, and then we describe how it works for SIMPS. We then move to quantum computation, arriving at the result that some SIMPS belonging to non-trivial SPT phases can be resources for universal MBQC on a single qubit. To the best of our knowledge, this gives the first example of a universal resource that does not have entanglement spectrum degeneracy.

\subsection{Quantum wire in MPS} \label{sec:mps_wire}

We first describe the notion of a \textit{quantum wire}. This is closely related to the \textit{localizable entanglement}, which describes the ability to concentrate entanglement between two distant spins in a many-body system by measuring the others in an optimal local basis \cite{Popp2005}. MPS give a very clear interpretation of localizable entanglement in terms of unitary evolution in the virtual space. Suppose that we measure each of the bulk spins of the MPS with OBC \eqref{eq:mps_obc}. This disentangles the bulk spins from the boundary spins, leaving behind the following state of the boundary spins,
\begin{equation} \label{eq:mps_obc_measured_eq}
    |\bar{\psi}'\rangle = \sum_{a,b}\   \langle a|A^{s_1}\cdots A^{s_N}|b\rangle |a,b\rangle,
\end{equation}
where $s_1,\dots, s_N$ are the measurement outcomes. This state can equivalently be written as,
\begin{equation}
    |\bar{\psi}'\rangle = (\id\otimes A^{s_1}\cdots A^{s_N})|\Psi^+_D\rangle
\end{equation}
where $|\Psi^+_D\rangle = \frac{1}{\sqrt{D}}\sum_{a=0}^{D-1}|aa\rangle$ is the Bell state.
We say that $|\bar{\psi}\rangle$ has \textit{long-range localizable entanglement} (LRLE) if the entanglement between the boundary spins in $|\bar{\psi}'\rangle$ remains finite as $N\rightarrow\infty$ \cite{Popp2005, Wahl2012}.  

One way to get LRLE is to choose the matrices $A^i$ to be unitary for all $i$. Then, the product $A^{s_1}\cdots A^{s_N}$ is also unitary, which implies that $|\bar{\psi}'\rangle$ is maximally entangled. In this case, the state in Eq.~\eqref{eq:mps_obc_measured_eq} can be viewed as a quantum computation in the virtual space. Reading from left to right, we begin with a logical state $|a\rangle$, apply a sequence of unitary operators $A^{s_1\dagger},\dots,A^{s_N\dagger}$, and then ``measure'' the logical state by projecting it onto $|b\rangle$. This unitary evolution picture explains why information can be faithfully transmitted along the chain via measurement \cite{Gross2007}.

The maximally entangled post-measurement state $|\bar{\psi}'\rangle$ depends on the measurement outcomes $s_1,\dots s_N$.
When using this state as a resource to, say, perform quantum teleportation, a correction unitary must be applied to the boundary degrees of freedom to map the state to a fixed reference state such as the Bell state. In this context, the operator $A^{s_1}\cdots A^{s_N}$ is called a \textit{byproduct operator} \cite{Raussendorf2003}. For generic states, the only way to undo the byproduct is to simply apply the unitaries $A^{s_i}$ one by one in reverse order, but this is not efficient in terms of depth, nor in terms of the number of classical bits that need to be communicated (which equals $\sim N \log d$ in this case). Ideally, one can instead perform some classical side processing on the measurement outcomes to determine a simple form of the byproduct operator. More precisely, we ask that only a finite number of classical bits are needed to specify the final byproduct.

If an injective MPS has LRLE and there is a classical side processing algorithm to identify the byproduct operator as described above, we say the MPS is a \textit{quantum wire}, since it can be used to deterministically and efficiently transmit quantum information encoded on one end to the other end using local measurements.

Another perspective on quantum wires comes from string order. As described in Sec.~\ref{sec:mps_spt}, MPS with OBC belonging to non-trivial SPT phases are symmetric under the string order parameters $V_g\otimes u_g^{\otimes N}\otimes V_g^*$ for $g\in G$ where $V_g$ and $u_g$ act on the boundary and bulk spins, respectively. For all abelian $G$, there is a basis that simultaneously diagonalizes $u_g$ for all $g\in G$ such that $u_g=\bigoplus_\alpha \chi^\alpha_g$ where $g\mapsto\chi^\alpha_g$ is a 1D representation of $G$. Suppose we measure all bulk spins in this basis. This measurement determines the symmetry charge $\chi^{\alpha_i}_g$ on every site $i$, and the total charge in the bulk is $\chi_g = \prod_{i=1}^N \chi^{\alpha_i}_g$. After measurement, the boundary spins are in a $+1$ eigenstate of the operators $\chi_g V_g\otimes V_g^*$. For certain projective representations $V_g$, these constraints are enough to uniquely specify the post-measurement state $|\bar{\psi}'\rangle$. Namely, suppose $V_g$ is an irreducible representation and $\Omega$ satisfies the ``maximally non-commutative'' condition \cite{Else2012}, which means $\Omega(g,h)=1$ for all $h$ if and only if $g=e$. Then these constraints imply that $|\bar{\psi}'\rangle = (I \otimes V_{g^*})|\Psi^+_D\rangle$ for a certain $g^*\in G$ that can be determined from $\chi_g$ with classical side processing \cite{Else2012}. Here, $V_{g^*}$ is the byproduct operator. A similar perspective appears in Ref.~\cite{Marvian2017}.

Hence, the string order parameters are crucial to the quantum wire property: they imply that measuring the bulk symmetry operators establishes long-range correlations between the boundaries, and the measured eigenvalues of the symmetry operators can be processed to determine the byproduct operators. This same intuition will hold for SIMPS, except that we will have to consider generalized string order corresponding to non-onsite symmetries.

\vspace{3mm}
\noindent
\textit{Example: the cluster state.} We can demonstrate the above ideas using the cluster state \eqref{eq:cluster_mps_X}.
Measuring all bulk spins in the $X$ basis gives the following state of the boundary spins,
\begin{equation} \label{eq:cluster_obc_measured}
    |\bar{C}'\rangle = \sum_{a,b}\   \langle a|HZ^{s_1}\cdots HZ^{s_N}|b\rangle |a,b\rangle
\end{equation}
where $s_i=0,1$ correspond to measurement outcomes $|+\rangle,|-\rangle$. Using $HZH = X$, the byproduct operator $HZ^{s_1}\cdots HZ^{s_N}$ can be rewritten, up to an unimportant global phase, as $X^{s_1+s_3+\dots + s_{N-1}}Z^{s_2+s_4+\dots + s_N}$, where we assume that $N$ is even. This byproduct operator can be described using two bits corresponding to the parity of measurement outcomes on all even sites and all odd sites. Therefore, the cluster state is a quantum wire.

As described in Sec.~\ref{sec:mps_spt}, the cluster state on open boundaries is a $+1$ eigenstate of the string operators $Z\otimes X_A \otimes Z$ and $X \otimes X_B \otimes X$. Now, suppose we measure all bulk spins in the $X$ basis. Then, the post-measurement state of the boundary spins is a $+1$ eigenstate of the operators $(-1)^a Z\otimes Z$ and $(-1)^bX\otimes X$ where $a=s_1+s_3+\dots + s_{N-1}$ is the measured eigenvalue of $X_A$ and $b=s_2+s_4+\dots + s_N$ is the measured eigenvalue of $X_B$. These relations tell us that the post-measurement state is $(\id\otimes X^aZ^b)|\Psi^+_2\rangle$, which matches the form of the byproduct operators derived from the MPS picture. 

\subsection{Robustness throughout the SPT phase} \label{sec:robustness_wire}

The universal fingerprints of the cluster phase \eqref{eq:else} imply that the quantum wire property is a robust property of the phase, and that the byproduct operators are uniform throughout the phase \cite{Else2012}. To see this, we can a choose a basis for the virtual space that respects the decomposition into the junk and logical subsystems, $\{|a\rangle\}_a = \{|a_j\rangle \otimes |a_\ell \rangle\}_a$. Then, Eq.~\eqref{eq:mps_obc_measured_eq} reads,
\begin{equation} \label{eq:junk_wire_1}
    |\bar{\phi}'\rangle = \sum_{a,b}\   \langle a|A_\phi^{s_1}\cdots A_\phi^{s_N}|b\rangle |a,b\rangle
    = |\bar{\phi}'_j\rangle \otimes |\bar{\phi}'_\ell\rangle,
\end{equation}
where,
\begin{equation} \label{eq:junk_wire_2}
    \begin{aligned}
        &|\bar{\phi}'_j\rangle = \sum_{a_j,b_j} \langle a_j|J_\phi^{s_1}\cdots J_\phi^{s_N}|b_j\rangle |a_j,b_j\rangle \\
        &|\bar{\phi}'_\ell\rangle = \sum_{a_\ell,b_\ell} \langle a_\ell|C_2^{s_1}\cdots C_2^{s_N}|b_\ell\rangle |a_\ell,b_\ell\rangle
    \end{aligned} 
\end{equation}
The state $|\bar{\phi}'_j\rangle$ is determined by the microscopic details of $|\phi\rangle$ and generically has entanglement that goes to 0 as $N\rightarrow\infty$. On the other hand, the state $|\bar{\phi}'_\ell\rangle$ is exactly the same state that results from measuring the cluster state, and is thus maximally entangled. Furthermore, the byproduct operators are the same as in the cluster state. Therefore, we can use the state $|\bar{\phi}\rangle$ as a quantum wire by simply using $|\bar{\phi}'_\ell\rangle$ for teleportation and throwing away the junk state $|\bar{\phi}'_j\rangle$.

\subsection{Quantum wire in SIMPS} \label{sec:simps_wire}

We now explain how quantum wire works in SIMPS. Starting with a SIMPS on OBC \eqref{eq:simps_obc} and measuring the bulk spins, the state of the remaining boundary spins is,
\begin{equation} \label{eq:simps_obc_measured}
    |\bar{\psi}'\rangle = \sum_{\alpha,\beta}\   \langle \alpha|B^{s_1s_2}\cdots B^{s_{N-1}s_N}|\beta\rangle |\alpha,\beta\rangle,
 \end{equation}
where $s_1,\dots, s_N$ are the measurement outcomes. Then, if the SIMPS tensors are unitary, as is the case for the states $|\psi_\ab\rangle$, this state describes a maximally entangled state between the $\chi$-dimensional boundary spins. Sometimes, the MPS tensors describing a given state are not unitary, but the SIMPS tensors will be. In such cases, the SIMPS representation gives a clearer understanding of origin of LRLE. This is discussed further in Sec.~\ref{sec:wahl}.

To study the form of the byproduct operators, we consider the example given in Eq.~\eqref{eq:simps_nice_ex}. If we measure the bulk of the OBC SIMPS described by the tensor in Eq.~\eqref{eq:simps_nice_ex}, the byproduct operator will be,
\begin{equation} \label{eq:simps_byproduct}
    \begin{aligned}
    &B^{s_1s_2}B^{s_2s_3}\cdots B^{s_{N-1}s_N} \\
    &\propto X^{s_1+s_2+\dots + s_{N-1}}Z^{s_1s_2+s_2s_3+\dots + s_{N-1}s_N}.
    \end{aligned}
\end{equation}
So, there is a simple classical side processing that can be performed to determine the byproduct operators. A similar result holds for all of the examples defined in Eq.~\eqref{eq:fp_simps}: the byproduct is some Pauli that can be simply determined using classical side processing. Therefore, these states are examples of quantum wires.

These examples have an important property that is distinct from the MPS examples described in Sec.~\ref{sec:mps_spt}. In the cluster state, and more general quantum wires emerging from SPT phases with onsite symmetries, the classical side processing to determine the byproduct operators is a linear function of the measurement outcomes. Conversely, Eq.~\eqref{eq:simps_byproduct} requires calculating a non-linear boolean function (\text{e.g.}, $s_1s_2$). This is significant, as all known schemes of teleportation (and, more generally, MBQC) have linear classical side processing relations\footnote{Possible exceptions are given by MBQC schemes that involve an initial round of measurements that reduce the resource state to a cluster state, followed by normal cluster state MBQC \cite{Wei2012,Miller2016}. However, since these protocols are non-deterministic, we do not consider them to be counterexamples.}. In fact, this linearity plays a key role in our fundamental understanding of MBQC, such as its relationship to quantum contextuality \cite{Raussendorf2013,Frembs2023}. Therefore, our examples may provide new fundamental insights into more general schemes of MBQC. 

We can also understand quantum wire in SIMPS from the perspective of string order. From Eq.~\eqref{eq:nice_symmetry}, we see that the SIMPS on OBC \eqref{eq:simps_obc} is invariant under,
\begin{equation}
    X\otimes\left( \prod_{i=1}^{N-1} CZ_{i,i+1}\right)\otimes X \quad \text{and} \quad Z\otimes \left(\prod_{i=1}^{N-1} Z_i\right) \otimes Z\ .
\end{equation}
When the bulk spins are measured in the $Z$ basis, the boundary spins are left in a state that is invariant under $(-1)^a X\otimes X$ and $(-1)^bZ\otimes Z$ where $b=s_1s_2+s_2s_3+\dots + s_{N-1}s_N$ is the measured eigenvalue of $\prod_i CZ_{i,i+1}$ and $a=s_1+s_2+\dots + s_{N-1}$ is the measured eigenvalue of $\prod_i Z_i$, implying that the state has the form $(\id\otimes X^bZ^a)|\Psi^+_2\rangle$. We see that the non-linearity of the classical side processing is a result of the non-onsite nature of the symmetries.

Finally, we remark that Theorem \ref{theorem:fingerprints} implies that the quantum wire property of the states $|\psi_\ab\rangle$ is shared by all states belonging to the corresponding SPT phases protected by the non-onsite symmetries, as shown in Sec.~\ref{sec:robustness_wire} for MPS.

\subsection{Measurement-based quantum computation with SIMPS} \label{sec:mbqc}

We have shown that SIMPS can act as quantum wires. In this section, we go one step further and show how they can also be used as resources for MBQC. In MBQC, unitary operators are applied to a logical state as it is teleported via measurement \cite{Raussendorf2003}. In this context, the quantum wire property corresponds to the ability to implement the identity gate in MBQC \cite{Else2012}.

For simplicity, we focus on the case of injective SIMPS with $L_1=1$ (see Definition \ref{def:normalsimps}). The protocol has two steps. First, we measure every odd numbered spin and obtain the measurement outcomes $\vec{s} = s_1,s_3,\dots, s_{N-1}$. The resulting state of the unmeasured even spins will be an MPS of the following form,
\begin{equation}
\begin{aligned}
    &|\psi(s_1,s_3,\dots,s_{N-1})\rangle = \\
    &\sum_{i_2,i_4,\dots,i_N}\mathrm{Tr}(A_{s_1,s_3}^{i_2} \cdots A_{s_{N-1},s_1}^{i_N})|i_2i_4\cdots i_N\rangle
\end{aligned}
\end{equation}
where $A_{s,s'}^i:=B^{si}B^{is'}$ is an injective MPS tensor for all $s,s'$ by assumption of injectivity of $B^{ij}$. After measuring half of the spins, the non-onsite symmetries of the state become ``disentangled'' into onsite symmetries of the unmeasured spins, where the form of these onsite symmetries varies from site to site depending on $\vec{s}$. The resulting MPS has SPT order under these symmetries. Therefore, we can apply existing techniques for using SPT states as resources for MBQC \cite{Else2012,Stephen2017}.

As an explicit example, we consider the following SIMPS with onsite dimension $d=4$,
\begin{equation} \label{eq:mbqc_simps}
    \begin{array}{llll}
        B^{00} = \id & B^{01} = X & B^{02} = \id & B^{03} = X \\[\medskipamount]
        B^{10} = \id & B^{11} = X  & B^{12} = \id & B^{13} = X \\[\medskipamount]
        B^{20} = Z & B^{21} = Y  & B^{22} = Y & B^{23} = Z \\[\medskipamount]
        B^{30} = Z & B^{31} = Y  & B^{32} = Y & B^{33} = Z
    \end{array}\ .
\end{equation}
This corresponds to the state $|\psi_\ab\rangle$ for a specific choice of $\ab$, chosen such that one of the $Z_2$ symmetries is non-onsite and the SIMPS is injective. By injectivity, we know that, for every $s,s'$, the matrices $A_{s,s'}^{i}$ span the space of $2\times 2$ matrices. Since each $A_{s,s'}^{i}$ is a Pauli, we find that for all $s,s'$, and for every Pauli $P=\id,X,Y,Z$, there is an index $i$ such that $A_{s,s'}^{i}=P$. These are the same matrices defining the MPS of the cluster state with a two-qubit unit cell, see Eq.~\eqref{eq:cluster_mps_X_2site}. Therefore, the state obtained after measuring every other site can be viewed as a cluster state with some permutations of basis states applied to every unit cell. Since these permutations can be determined from $\vec{s}$, to which we have direct access, we can undo it and then apply the usual scheme of universal MBQC using cluster states \cite{Raussendorf2003}. Therefore, the SIMPS in Eq.~\eqref{eq:mbqc_simps} is a  resource for universal MBQC on a single logical qubit. 

Remarkably, the SIMPS in Eq.~\eqref{eq:mbqc_simps} does not have a degenerate entanglement spectrum. More specifically, the entanglement spectrum using the techniques of Ref.~\cite{Cirac2011} can be computed to be $\approx(0.43,0.25,0.25,0.07)$, and the degeneracy of the second and third eigenvalues is not symmetry-protected.
To the best of our knowledge, this is the first example of a universal resource for MBQC without entanglement spectrum degeneracy.\footnote{One other example of a quantum wire without entanglement spectrum degeneracy was given in Ref.~\cite{Liu2024}, but this example possesses long-range entanglement, which disqualifies it as a quantum wire by our definition as its preparation is as difficult a task as teleportation itself.} Similar MBQC protocols can be constructed for all states $|\psi_\ab\rangle$, but whether or not we get a universal gate set depends on $\ab$.

\section{Constrained spin chains} \label{sec:constrained}

In this section, we study constrained spin chains which are defined by a Hilbert space that is not a tensor product of onsite Hilbert spaces. We find that SIMPS provide a very simple way to encode these constraints, and show how previous examples in the literature can be cast into a SIMPS language.

\subsection{Variational SIMPS for Rydberg chain}
As a concrete example, we will consider the constrained Rydberg chains \cite{Turner2018} which consist of two states per site, corresponding to the ground and excited states of an atom. The excited state, called a Rydberg state, has a large principal quantum number, and the large spatial extent of the Rydberg state means that neighbouring atoms in the Rydberg state will experience a large energy penalty, a phenomenon known as the Rydberg blockade \cite{Lukin2001}. A good approximation of the physics of these chains is obtained by considering only those states in which there are no neighbouring atoms in the Rydberg state. This corresponds to a constrained Hilbert space that is not a tensor product of local Hilbert spaces on each site. The simplest dynamics in this system are captured by the so-called PXP model \cite{Turner2018}.

We first consider how such a constrained spin chain can be captured in MPS. Referring to the ground and Rydberg states as $|0\rangle$ and $|1\rangle$, the constraint can be implemented by restricting to MPS tensors $A^i$ satisfying 
\begin{equation} \label{eq:mps_rydberg}
A^1A^1=0\ .
\end{equation}
Then, any component of the MPS wavefunction with two neighbouring Rydberg atoms is automatically equal zero. In fact, if the MPS is normal, then the injectivity of the map in Eq.~\eqref{eq:inj_map} implies that this is the only way to enforce the constraint. Notice that Eq.~\eqref{eq:mps_rydberg} implies that $A^1$ is not full rank. This immediately leads to the following result,
\begin{Prop} \label{theorem:rydberg_good}
    If a normal MPS satisfies the Rydberg constraint, then it makes a good SIMPS. A similar conclusion holds for more general local constraints.
\end{Prop}
\noindent

By converting to SIMPS, we can implement the constraint in a simpler manner. Namely, we simply restrict to tensors $B^{ij}$ such that 
\begin{equation} \label{eq:simps_rydberg}
    B^{11}=0\ .
\end{equation}
No further constraints are necessary. We can similarly argue that this is the only way to encode the constraint for normal SIMPS, so it is a necessary and sufficient condition.
Unlike Eq.~\eqref{eq:mps_rydberg} which is a \textit{nonlinear} constraint on the entries of the MPS tensor, Eq.~\eqref{eq:simps_rydberg} is a \textit{linear} constraint on the entries of the SIMPS tensor. The simpler implementation of the constraint in SIMPS could be useful in variational algorithms. For MPS, one would need to optimize the tensor $A^i$ while enforcing the nonlinear constraint $A^1A^1=0$. Conversely, for SIMPS, one need only set $B^{11}=0$ and optimize the other components freely.

To explicitly demonstrate the use of SIMPS for constrained spin chains, we reexamine previous works where MPS were used to study Rydberg chains. Namely, Refs.~\cite{Bernien2017,Ho2019,Michailidis2020,Ljubotina2022,Li2023} considered MPS tensors of the form,
\begin{equation}
    A^0 =\begin{pmatrix}
        a & 0  \\
        a & 0  
    \end{pmatrix}
    ,\quad
     A^1 =\begin{pmatrix}
        0 & b  \\
        0 & 0  
    \end{pmatrix}\ ,
\end{equation}
for some coefficients $a,b$. Clearly, $A^1A^1=0$, so the states generated by these tensors satisfy the Rydberg constraint for all $a,b$. Both $A^0$ and $A^1$ have rank equal to one, so they should be representable as bond dimension $\chi=1$ SIMPS. Indeed, we can represent the same states with the following SIMPS tensor,
\begin{equation}
    \begin{array}{ll}
        B^{00}=a & B^{01} = b   \\[\medskipamount]
        B^{10}=a & B^{11} = 0 
    \end{array}
\end{equation}
Using the redundancy provided by SIMPS gauge transformations \eqref{eq:fundamental}, one can show that this is in fact the most generic bond dimension $\chi=1$ SIMPS which satisfies $B^{11}=0$. By considering SIMPS of higher bond dimension satisfying $B^{11}=0$, we can systematically go beyond this two-parameter family of states in the space of Rydberg-constrained states.

\subsection{SIMPS for an exact scar state}
As a second example of the use of SIMPS for constrained spin chains, we consider AKLT state \cite{Affleck1989}. The state is defined on a spin-1 Hilbert space whose states we label as $\{|0\rangle, |{\uparrow}\rangle, |{\downarrow}\rangle\}$. While not usually considered as a constrained state, the AKLT state does have ``hidden'' antiferromagnetic order \cite{Kennedy1992,Else2013} which manifests in the fact that the patterns $|{\uparrow}\rangle|{\uparrow}\rangle$ and $|{\downarrow}\rangle|{\downarrow}\rangle$ never appear on two neighbouring spins (along with longer range constraints). Additionally, it was shown in Ref.~\cite{Lin2019} how to map the AKLT state onto an exact scar state of the PXP model, where one of the above constraints maps onto the Rydberg constraint.

Up to constant factors, the MPS of the AKLT state is,
\begin{equation} \label{eq:aklt_mps}
    A^0 = Z,\quad A^{\uparrow} = |0\rangle\langle 1|,\quad A^{\downarrow} = |1\rangle\langle 0|
\end{equation}
Notice that $A^0$ is rank-2, but $A^{\uparrow}$ and $A^{\downarrow}$ are rank-1. Therefore, the SIMPS has mixed bond dimension $\chi^0 = 2$ and $\chi^{\uparrow}=\chi^{\downarrow}=1$. Using Eq.~\eqref{eq:mpstosimps}, we find,
\begin{equation} \label{eq:aklt_simps}
    \begin{array}{lll}
        B^{00}=Z & B^{0{\uparrow}} = |0\rangle & B^{0{\downarrow}} = |1\rangle  \\[\medskipamount]
        B^{{\uparrow}0}=-\langle 1| & B^{{\uparrow}{\uparrow}}= 0 & B^{{\uparrow}{\downarrow}} = 1 \\[\medskipamount]
        B^{{\downarrow}0}=\langle 0| & B^{{\downarrow}{\uparrow}} = 1 & B^{{\downarrow}{\downarrow}} = 0
    \end{array}\ .
\end{equation}
Observe that $B^{{\uparrow}{\uparrow}}=B^{{\downarrow}{\downarrow}}=0$ as required by the constraint. Using this form, we can easily deform the AKLT state while preserving the constraint by fixing $B^{{\uparrow}{\uparrow}}=B^{{\downarrow}{\downarrow}}=0$ and freely varying the other components. {This could facilitate, \textit{e.g.}, finding an exact path of SIMPS in the constrained Hilbert space that goes through a phase transition, similar to the MPS paths derived in Refs.~\cite{Wolf2006,Jones2021}.}

Going forward, we believe the simpler implementation of local constraints in SIMPS as compared to MPS will be generally useful for studying constrained spin chains both numerically and analytically, as we discuss further in the next section.

\section{Discussion} \label{sec:discussion}

We finish by discussing our results. First, we use our results to analyze an example that appeared previously in the literature. Then, we discuss our results in the context of a conjectured equivalence between quantum wire and SPT order. Finally, we suggest some future directions of research for SIMPS.

\subsection{Analyzing an example from Ref.~\cite{Wahl2012}} \label{sec:wahl}

Here we show how SIMPS can clarify some features of an example that appeared previously in the literature. The example was given in Ref.~\cite{Wahl2012}, as an example of an MPS with LRLE. It is defined by the following MPS with $d=D=3$,
\begin{equation}
\begin{aligned} \label{eq:wahl_tensor}
    &A^0 =\begin{pmatrix}
        1 & 0 & 0 \\
        0 & 1 & 0 \\
        1 & 0 & 0
    \end{pmatrix}
    ,\quad
    A^1 =\begin{pmatrix}
        0 & 0 & 1 \\
        0 & 1 & 0 \\
        0 & 0 & -1
    \end{pmatrix},
    \\
    &A^2 =\begin{pmatrix}
        0 & 1 & 0 \\
        1 & 0 & 0 \\
        0 & 1 & 0
    \end{pmatrix}.
\end{aligned}
\end{equation}
It was shown in Ref.~\cite{Wahl2012} that, by measuring the bulk spins, one can establish a maximally entangled state between a pair of two-dimensional spins that are appropriately coupled to the boundaries. This example initially appears puzzling for a number of reasons. First, there is LRLE despite the fact that the MPS tensors are not unitary, and the amount of entanglement generated is that of a pair of two-dimensional spins rather than three-dimensional as one may have expected from a $D=3$ MPS. Second, it is not clear if the state is a quantum wire, \text{i.e.}, if there is an efficient classical side processing to determine the byproduct operator. Finally, Ref.~\cite{Wahl2012} showed that this state only has a single onsite $Z_2$ symmetry, which is not enough to support non-trivial SPT order. Therefore, it is not clear if the LRLE can be understood in terms of SPT order and string order parameters. 

SIMPS provide the answer to all of these questions and more. Using Eq.~\eqref{eq:mpstosimps}, we find that this state can be represented by the following SIMPS,
\begin{equation} \label{eq:simps_ex_wahl}
    \begin{array}{lll}
        B^{00}=\id & B^{01} = X & B^{02} = \id  \\[\medskipamount]
        B^{10}=X & B^{11} = Z & B^{12} = X \\[\medskipamount]
        B^{20}=X & B^{21} = \id & B^{22} = X
    \end{array}\ .
\end{equation}
This is a state of the form $|\psi_\ab\rangle$, so all of our analysis of the physical and computational properties of these states applies. In particular, the LRLE can be attributed to unitary evolution in the SIMPS virtual space, the amount of entanglement generated matches the SIMPS bond dimension $\chi=2$, and the state is a quantum wire. The state also has a pair of non-onsite symmetries defined by Eq.~\eqref{eq:ab_symms} and belongs to a non-trivial SPT phase protected by them, and the quantum wire follows from measuring these string order parameters. Finally, our analysis shows that the quantum wire property of this state is robust to perturbations that preserve the non-onsite symmetries.

Ref.~\cite{Wahl2012} also constructed the most general form of MPS having LRLE, with Eq.~\eqref{eq:wahl_tensor} being one example.  The authors showed that LRLE in MPS can always be understood in terms of an underlying unitary evolution in the virtual space of a different tensor network. However, the general form of this tensor network is very non-local, and is not helpful in assessing the symmetries of a state, nor its capability to act as a quantum wire. In a sense, SIMPS are a bridge between MPS and the general structure outlined in Ref.~\cite{Wahl2012}. It would be interesting to study whether the general MPS of Ref.~\cite{Wahl2012} are additionally quantum wires, and what symmetries they possess.

\subsection{Quantum wire and SPT order}

Our results support the conjecture that all quantum wires must have string order. This connection was first explored in Ref.~\cite{Verstraete2004} and later in Refs.~\cite{Popp2005,Venuti2005,Doherty2009,Skrovseth2009,Hong2023}. In Ref.~\cite{Popp2005}, a counterexample of a state with LRLE but no string order was proposed. However, that state is not a quantum wire as it does not admit an efficient way to compute the byproduct operator, so it does not refute the conjecture.

We formally state the conjecture here,
\begin{Conj} \label{conj}
    Any 1D state that possess long-range localizable entanglement with an efficient method of computing byproduct operators must belong to a non-trivial SPT phase.
\end{Conj}
\noindent
We are purposely vague in our meaning of ``efficient'' as different definitions can lead to different results. For example, under the assumption that all classical side processing of measurement outcomes is linear and byproduct operators are Pauli operators, Ref.~\cite{Hong2023} showed that teleportation implies SPT order with onsite symmetries. In contrast, the results herein violate the first assumption, since we needed non-linear classical processing to use our models for teleportation. Nevertheless, our models still have SPT order, albeit only under a generalized notion of symmetry (non-onsite symmetry). Our results therefore support Conjecture \ref{conj}, while also showing how the notion of SPT order can and should be generalized. We expect that relaxing these assumptions further may lead to even more general notions of symmetry and SPT order, which should have interesting physical and computational ramifications.

\subsection{Future directions}


There are several interesting applications for which SIMPS may be useful. First, we have shown how SIMPS are suited to capture all classes of anomalous group-like symmetries (see Sec.~\ref{sec:anomalous}). Ref.~\cite{GarreRubio2023} recently derived a classification of phases of matter with such symmetries in terms of so-called ``$L$-symbols''.
Is it possible to reproduce this classification and reinterpret the $L$-symbols in terms of SIMPS gauge transformations, similar to how the classification of phases with onsite symmetries uses MPS gauge transformations?

As discussed in Sec.~\ref{sec:constrained}, we expect that SIMPS will be useful in numerical studies of constrained spin chains due to the linear implementation of the constraint. However, making use of this constraint would require modifying standard algorithms such that the SIMPS structure is properly taken advantage of (\text{i.e.}, simply rewriting the SIMPS as an MPS and doing normal MPS algorithms would get rid of the possible advantages). Whether the cost of this modification outweighs the benefits of the linear implementation of the constraint remains an open question. We also hope that SIMPS will be useful for analytical studies of constrained spin chains. For example, they may help to classify phases of matter in constrained spin chains.

Finally, an obvious extension of SIMPS is to go to higher dimensions. We can, for example, define two-dimensional tensor networks where each tensor depends on the value of all spins around a plaquette of the lattice. This will be reminiscent of the plaquette entanglement structure of the CZX model \cite{Chen2011a} which fits into the general framework of semi-injective PEPS \cite{Molnar2018semi}. These states would enable the study of non-onsite symmetries, quantum computation, and constrained systems in higher dimensions, as we have done here for 1D.

\acknowledgements

We thank A. Friedman, O. Hart, Y. Hong, N. Schuch, R. Vanhove, and R. Verresen for helpful comments. This work was funded by the Simons Collaboration on Ultra-Quantum Matter (UQM), which is funded by grants from the Simons Foundation (651440).

\printbibliography

\appendix 

\section{Proof of Proposition \ref{prop:conversion}} \label{app:conversion}

In this section we prove Proposition \ref{prop:conversion}. Consider a normal SIMPS tensor $B^{ij}$, and define the associated MPS tensor $A^i$ as in Eq.~\eqref{eq:simpstomps}. We have,
\begin{equation}
    A^{i_1}\cdots A^{i_L} = w^\dagger (|i_1\rangle\langle i_L|\otimes B^{i_1i_2}\cdots B^{i_{L-1}i_L})u
\end{equation}
where we used the fact that $\langle i|uw^\dagger|j\rangle=B^{ij}$. Given an arbitrary $D\times D$ matrix $M$, define the $d\chi\times d\chi$ matrix $M'=wMu^\dagger$. Then, decompose $M'$ as,
\begin{equation}
    M'=\sum_{s_1,s_2=0}^{d-1}|s_1\rangle\langle s_2|\otimes M_{s_1,s_2}\ ,
\end{equation}
where each $M_{s_1,s_2}$ is a $\chi\times \chi$ matrix. Since $B^{ij}$ is normal, $M_{s_1,s_2}$ is contained in $\mathrm{span}\{ B^{s_1i_1}\cdots B^{i_{L}{s_2}}\}_{i_1,\dots i_L}$ for $L\geq L_1$, so $M'$ is contained in $\mathrm{span}\{ |i_1\rangle\langle i_L|\otimes B^{i_1i_2}\cdots B^{i_{L-1}i_L}\}_{i_1,\dots, i_L}$ for $L\geq L_1+2$. Therefore, $w^\dagger M' u=M$ is in $\mathrm{span}\{A^{i_1}\cdots A^{i_L}\}_{i_1,\dots,i_L}$ for $L\geq L_1+2$. Since $M$ was arbitrary, $A^i$ is normal with injectivity length $L_0\leq L_1+2$. 

Now, consider a normal MPS tensor $A^i$ and define the associated SIMPS tensor $B^{ij}$ as in Eq.~\eqref{eq:mpstosimps}. Then we have,
\begin{equation}
\begin{aligned}
    &B^{s_1i_1} B^{i_1i_2}\cdots B^{i_{L-1}i_L}B^{i_Ls_2} = \\
    &P^{s_1\dagger}A^{i_1}A^{i_2}\cdots A^{i_L}A^{s_2}P^{s_2}
\end{aligned}
\end{equation}
Note that $A^{s_2}$ is invertible on its domain. That is, there is a matrix $(A^{s_2})^{-1}$ such that $(A^{s_2})^{-1}A^{s_2}P^{s_2}=P^{s_2}$. Now, given any $\chi^{s_1}\times \chi^{s_2}$ matrix $M$, define the $D\times D$ matrix $M'=P^{s_1}MP^{s_2\dagger}(A^{s_2})^{-1}$. Since $A^i$ is normal, $M'$ is contained in $\mathrm{span}\{A^{i_1}A^{i_2}\cdots A^{i_{L}}\}_{i_1,\dots,i_{L}}$ when $L\geq L_0$. Therefore, $P^{s_1\dagger}M'A^{s_2}P^{s_2}=M$ is contained in $\mathrm{span}\{B^{s_1i_1} B^{i_1i_2}\cdots B^{i_{L-1}i_L}B^{i_Ls_2}\}_{i_1,\dots,i_L}$ when $L\geq L_0$. Since $M$ was arbitrary, $B^{ij}$ is normal with injectivity length $L_1\leq L_0$. \hfill \qedsymbol

\section{Proof of Theorem \ref{theorem:fundamental}} \label{app:fundamental}

Here we prove the fundamental theorem of SIMPS. Suppose $B^{ij}$ and $A^{ij}$ are two normal SIMPS tensors with injectivity length $L_1$ that generate the same state for all system sizes $N$. For simplicity of the presentation, we take $L_1=1$, but a similar proof follows for arbitrary $L_1$. Take $N$ even with $N\geq 6$ and consider the following states obtained by projecting every other spin onto a fixed state $|s_i\rangle$,
\begin{equation}
\begin{aligned}
    &|\psi_A(s_1,s_3,\dots,s_{N-1})\rangle = \\
    &\sum_{i_2,i_4,\dots,i_N}\mathrm{Tr}(A^{s_1i_2}A^{i_2s_3}\cdots A^{i_Ns_1})|i_2i_4\cdots i_N\rangle
\end{aligned}
\end{equation}
and $|\psi_B(s_1,s_3,\dots,s_{N-1})\rangle$ is defined similarly with $A^{ij}\rightarrow B^{ij}$. By assumption, $|\psi_A(s_1,s_3,\dots,s_{N-1})\rangle=|\psi_B(s_1,s_3,\dots,s_{N-1})\rangle$ for all $s_i$. Observe that $|\psi_A(s_1,s_3,\dots,s_{N-1})\rangle $ is a non-translationally invariant MPS defined by the injective tensors $A^{s_j,i}A^{i,s_{j+2}}$ around site $j$ and similarly for $B$. Therefore, by the Fundamental Theorem of MPS, as stated in Theorem 1 of Ref.~\cite{Molnar2018}, there exist invertible matrices $X_{s_j,s_{j+2}}$ such that 
\begin{equation}
A^{s_j,i}A^{i,s_{j+2}}=X_{s_j,s_{j+2}}B^{s_j,i}B^{i,s_{j+2}}X_{s_{j+2},s_{j+4}}^{-1}.
\end{equation}
The right-hand side of the above equation depends on $s_{j+4}$ while the left-hand side does not. Using this and the fact that $B^{ij}$ is normal, we can show that $X_{s,s'}$ is independent of $s'$. Therefore, we can write $X_{s,s'}\equiv X_{s}$ such that,
\begin{equation} \label{eq:ft1}
A^{s,i}A^{i,s'}=X_{s}B^{s,i}B^{i,s'}X_{s'}^{-1}.
\end{equation}
We can repeat the above calculation with $N=3k$ and $k\geq 3$ where every third spin is projected onto a fixed state. Doing so, we find that,
\begin{equation} \label{eq:ft2}
A^{s,i}A^{i,j}A^{j,s'}=\widetilde{X}_{s}B^{s,i}B^{i,j}B^{j,s'}\widetilde{X}_{s'}^{-1}.
\end{equation}
for some other invertible matrices $\widetilde{X}_{s}$. By concatenating Eq.~\eqref{eq:ft1} three times and comparing to Eq.~\eqref{eq:ft2} concatenated two times, we can use the uniqueness of gauge transformations to conclude that $X_s\propto \widetilde{X}_{s}$ where the implied proportionality constant does not depend on $s$. Finally, we have,
\begin{equation}
    \begin{aligned}
    A^{s,i}A^{i,j}A^{j,s'} &={X}_{s}B^{s,i}B^{i,j}B^{j,s'}{X}_{s'}^{-1} \\
    &=X_{s}B^{s,i}B^{i,j}X_j^{-1}A^{j,s'}
    \end{aligned}
\end{equation}
If we set $j=s$, and use the fact that $B^{i,j}$ is normal, we can contract the index $i$ to replace $B^{s,i}B^{i,s}$ with $I$, giving,
\begin{equation}
    {X}_{s}B^{s,s'}{X}_{s'}^{-1}
    =X_{s}X_{s}^{-1}A^{s,s'}=A^{s,s'}
\end{equation}
which is the claimed result. \hfill \qedsymbol

\section{Proof of Lemma \ref{lemma:fingerprints}} \label{app:fingerprints}

The following proof applies to any normal SIMPS of the form given in Eq.~\eqref{eq:fp_simps}. First, we present some facts about the symmetries of the SIMPS. The matrices $B_\ab^{ij}$ satisfy the relations,
\begin{equation} \label{eq:fp_simps_symms_1a}
    \begin{aligned}
        &(-1)^{b_{ij}}B_\ab^{ij} = XB_\ab^{ij}X \\
        &(-1)^{a_{ij}}B_\ab^{ij} = ZB_\ab^{ij}Z \\
    \end{aligned}
\end{equation}
as well as,
\begin{equation} \label{eq:fp_simps_symms_2}
    \begin{aligned}
        &B_\ab^{i+1,j} = X^{a_{ij}+a_{i+1,j}}B_\ab^{ij}Z^{b_{ij}+b_{i+1,j}} \\
        &B_\ab^{i,j+1} = X^{a_{ij}+a_{i,j+1}}B_\ab^{ij}Z^{b_{ij}+b_{i,j+1}} 
    \end{aligned}
\end{equation}
where all addition on indices is done modulo $d$. As discussed in the main text, Eq.~\eqref{eq:fp_simps_symms_1a} implies a group of non-onsite symmetries $G_\ab=\langle U^{\bf{a}}, U^{\bf{b}}\rangle$. If we assume the SIMPS is normal, then these symmetries form a faithful representation of $Z_2\times Z_2$, \text{i.e.}, $U^{\bf{a}}$ and $U^{\bf{b}}$ are both non-trivial operators and they are distinct from each other, even after modding out by phase factors that are independent of $\ab$. To see this, recall that $B_\ab^{ij}$ being normal means that there exists an $L$ such that $\mathrm{span}\{B_\ab^{s_1i_1}B_\ab^{i_1i_2}\cdots B_\ab^{i_{L_0}s_2}\}_{i_1,\dots, i_{L_0}}=M_\chi$ for all $s_1,s_2$, where $M_\chi$ is the space of all $\chi\times\chi$ matrices. In the present case, the product $B_\ab^{s_1i_1}B_\ab^{i_1i_2}\cdots B_\ab^{i_{L_0}s_2}$ is proportional to $X^{\alpha_{s_1,i_1,i_2,\dots,i_{L_0},i_{s_2}}} Z^{\beta_{s_1,i_1,i_2,\dots,i_{L_0},i_{s_2}}}$ where, 
\begin{align}
\alpha_{i_1,\dots, i_n}&={a_{i_1i_2}+\dots+a_{i_{n-1}i_n}}\mod 2, \nonumber\\
\beta_{i_1,\dots, i_n}&={b_{i_1i_2}+\dots+b_{i_{n-1}i_n}}\mod 2. \nonumber
\end{align}
These products span the whole space of matrices if there is a choice of $i_1,\dots, i_{L_0}$ such that the above product equals $P$ for all $P=I,X,Y,Z$, which means the map, 
\begin{equation} \label{eq:surj_map}
\Gamma:i_1,\dots, i_{L_0}\mapsto (\alpha_{s_1,i_1,i_2,\dots,i_{L_0},i_{s_2}},\beta_{s_1,i_1,i_2,\dots,i_{L_0},i_{s_2}}),
\end{equation}
is a surjective map into the set $\{(x,y)|x,y=0,1\}$ for all $s_1,s_2$. 
Finally, we have, 
\begin{align}
U^{\bf{a}}|i_1\cdots i_n\rangle&=(-1)^{\alpha_{i_1,\dots, i_n}}|i_1\cdots i_n\rangle, \nonumber\\
U^{\bf{b}}|i_1\cdots i_n\rangle&=(-1)^{\beta_{i_1,\dots, i_n}}|i_1\cdots i_n\rangle. \nonumber
\end{align}
If the symmetry is not faithful (up to phase factors), then one of $U^{\bf{a}}$, $U^{\bf{b}}$, or $U^{\bf{a}}U^{\bf{b}}$ is proportional to the identity. But this cannot be the case due to the surjectivity of the map $\Gamma$. For example, if $U^{\bf{a}}=\lambda I$ for some $\lambda$, then $\alpha_{i_1,\dots, i_n}$ is independent of $i_1,\dots, i_n$, which contradicts the surjectivity of $\Gamma$ if $n\geq L_0+1$. A similar conclusion holds for the other group elements. Therefore, the symmetry must be faithful for $n\geq L_0+1$.

Now we are ready to prove Lemma 1. The goal of the proof is to use the relations in Eq.~\eqref{eq:fp_simps_symms_1a} and \eqref{eq:fp_simps_symms_2} to replace any symmetric operator $O$ with a diagonal operator $O'$ which acts the same way as $O$ on $|\psi_\ab\rangle$.
Consider any operator $O$ supported on the interval $[x,y]$. This operator can be decomposed in the following way,
\begin{equation}
    O = \sum_{\vec{j}} \mathcal{D}_{\vec{j}}\mathcal{X}_{\vec{j}}.
\end{equation}
Where $\vec{j}=(j_x,j_{x+1},\dots,j_{y})$ is a vector with $j_k =0,\dots, d-1$, $\mathcal{X}_{\vec{j}} = \bigotimes_{k=x}^y \mathcal{X}_k^{j_k}$, $\mathcal{X}=\sum_{i=0}^{d-1} |i+1\rangle\langle i|$ is the generalized Pauli-$X$ operator, and $\mathcal{D}_{\vec{j}}$ is a diagonal matrix.
Note that each operator $O_{\vec{j}}:=\mathcal{D}_{\vec{j}}\mathcal{X}_{\vec{j}}$ is linearly independent. Suppose $[O,U^{\bf{a}}]=[O,U^{\bf{b}}]=0$. Then, since $[O_{\vec{j}},U^{\bf{a}}]\propto O_{\vec{j}}$ (and similarly for $U^{\bf{a}}$), it must be true that $[O_{\vec{j}},U^{\bf{a}}]=[O_{\vec{j}},U^{\bf{b}}]=0$ for each $\vec{j}$ individually. From now on, assume $x=1$ WLOG.

Consider the action of $\mathcal{X}_{\vec{j}}$ on $|\psi_\ab\rangle$. From Eq.~\eqref{eq:fp_simps_symms_2}, we can always push the action of any product of $\mathcal{X}$ operators onto the virtual level of the SIMPS as some product of Pauli operators. Furthermore, using Eq.~\eqref{eq:fp_simps_symms_1a}, we can always move these Pauli operators around the virtual space by acting on the physical degrees of freedom with the operators $u^{\bf{a}}$ and $u^{\bf{b}}$. Because of this, we can collect all of the Pauli operators in one position in the trace in Eq.~\eqref{eq:simps}. Collecting them around site $x=1$, we have,
\begin{equation}
    O_{\vec{j}}|\psi_\ab\rangle = \mathcal{D}_{\vec{j}}\sum_{\vec{i}} s_{\vec{i},\vec{j}} \mathrm{Tr} (P_{\vec{i},\vec{j}}B^{i_1i_2}\cdots B^{i_Ni_1})|\vec{i}\rangle
\end{equation}
where we use the shorthand notation $\vec{i}=i_1,\cdots, i_N$. Therein, $P_{\vec{i},\vec{j}}$ is some Pauli operator and $s_{\vec{i},\vec{j}}=\pm 1$ is a result of using  Eq.~\eqref{eq:fp_simps_symms_1a} to move the Paulis around. We can absorb the signs $s_{\vec{i},\vec{j}}$ into $\mathcal{D}_{\vec{j}}$ by defining a modified diagonal operator $\mathcal{D}'_{\vec{j}}$. Note that the signs $s_{\vec{i},\vec{j}}$ depend only on the values of $\vec{i}$ in the interval $[x,y]$, so $\mathcal{D}'_{\vec{j}}$ is still supported on the interval $[x,y]$. 

Then we have,
\begin{equation} \label{eq:proof_twisted_state}
    O_{\vec{j}}|\psi_\ab\rangle = \mathcal{D}'_{\vec{j}}\sum_{\vec{i}} \mathrm{Tr} (P_{\vec{i},\vec{j}}B^{i_1i_2}\cdots B^{i_Ni_1})|\vec{i}\rangle.
\end{equation}
Now we use the fact that $O$ is a symmetric operator. This implies that $U^{\bf{a}}O_{\vec{j}}|\psi_\ab\rangle = O_{\vec{j}}|\psi_\ab\rangle$, and similarly for $U^{\bf{b}}$, \text{i.e.}, the state $O_{\vec{j}}|\psi_\ab\rangle$ has neutral symmetry charge. Up to $\mathcal{D}'_{\vec{j}}$, Eq.~\eqref{eq:proof_twisted_state} can be interpreted as a symmetry twisted state, defined by the insertion of a virtual symmetry representation into the SIMPS, as discussed in Sec.~\ref{sec:twists}. The twisted state has neutral charge if and only if $P_{\vec{i},\vec{j}}=I$ for all $\vec{i},\vec{j}$ such that $\mathcal{D}'_{\vec{j}}|\vec{i}\rangle\neq 0$ (note that this only works if $U^{\bf{a}}$ and $U^{\bf{b}}$ are distinct non-trivial operators). Therefore, we can remove $P_{\vec{i},\vec{j}}$ from Eq.~\eqref{eq:proof_twisted_state} without changing the state, giving,
\begin{equation}
    O_{\vec{j}}|\psi_\ab\rangle = \mathcal{D}'_{\vec{j}}\sum_{\vec{i}} \mathrm{Tr}(B^{i_1i_2}\cdots B^{i_Ni_1})|\vec{i}\rangle \equiv \mathcal{D}'_{\vec{j}}|\psi_\ab\rangle
\end{equation}
So the operators $O$ and $O':=\sum_{\vec{j}} \mathcal{D}'_{\vec{j}}$ act in the same way on $|\psi_\ab\rangle$. Note that $O'$ is diagonal, and that it is supported on the same interval as $O$. \hfill \qedsymbol

\end{document}